\documentclass[USenglish,oneside,twocolumn]{article}

\usepackage[utf8]{inputenc}
\usepackage[big]{dgruyter_NEW}
 \usepackage[table]{xcolor}
\usepackage{epsfig,endnotes}
\usepackage{varwidth}
\usepackage{algorithm}
\usepackage{algpseudocode}
\usepackage{balance}
\usepackage{color}
\usepackage{nicefrac}
\usepackage{amssymb,amsmath,amsthm}
\usepackage{hyperref}

\usepackage{bm}
\usepackage{mathtools}
\usepackage{booktabs}
\usepackage{multirow}
\usepackage{bigdelim}
\usepackage{mathtools}
\usepackage{indentfirst}
\usepackage{siunitx}
\usepackage{enumitem}
\usepackage{boxedminipage}
\usepackage{float}
\usepackage{graphicx}
\usepackage{caption}
\usepackage{subcaption}
\usepackage[normalem]{ulem}
\usepackage{xpatch}
\usepackage{stmaryrd}
\usepackage{xspace}
\usepackage{placeins}
\usepackage{pifont}
\usepackage{wasysym}

\newif\iffullversion
\fullversionfalse

\usepackage{adjustbox}
\usepackage{array}
\newcolumntype{R}[2]{%
    >{\adjustbox{angle=#1,lap=\width-(#2)}\bgroup}%
    l%
    <{\egroup}%
}
\newcommand*\rot{\multicolumn{1}{R{70}{1em}}}% no optional argument here, please!

\newcommand{\mc}[2]{\multicolumn{#1}{c}{#2}}
\newcommand{\mr}[2]{\multirow{#1}{*}{#2}}

\makeatletter
\def\plist@algorithm{Algorithm\space}
\makeatother

\newcommand{\cmark}{\CIRCLE}
\newcommand{\xmark}{\Circle}
\newcommand{\smark}{\LEFTcircle}
\def\BibTeX{{\rm B\kern-.05em{\sc i\kern-.025em b}\kern-.08em
    T\kern-.1667em\lower.7ex\hbox{E}\kern-.125emX}}

\newtheorem{theorem}{Theorem}

\newcommand{\ThisWork}{{\sc Falcon}}
\newcommand{\falcon}{{\sc Falcon~}}
\newcommand{\comment}[1]{}

\newcommand{\ch}[1]{{\color{black} #1}}
\newcommand{\chh}[1]{{\color{black} #1}}

\graphicspath{{./Images/}}
\DeclareGraphicsExtensions{.pdf,.jpg,.png}
\algnewcommand\algorithmicinput{\textbf{Common Randomness:}}
\algnewcommand\CommonR{\item[\algorithmicinput]}

\newcommand{\tinyth}{$^{\mathsf{{\tiny th}}}$}
\newcommand{\party}[1]{P_{#1}}
\newcommand{\prm}{p}
\newcommand{\ring}{L}

\newcommand{\share}[1]{\llbracket #1\rrbracket}
\newcommand{\Share}[2]{\llbracket #1\rrbracket^{#2}}
\newcommand{\protx}[1]{\ensuremath{\Pi_{\mathsf{#1}}}\xspace}
\newcommand{\wa}[1]{\mathsf{wrap}_{#1}}

\newcommand{\bbZ}{\mathbb{Z}}
\newcommand{\zo}{\{0,1\}}

\definecolor{dblue}{rgb}{0.00, 0.50, 0.90}
\definecolor{lblue}{rgb}{0.70, 0.80, 1.00}
\definecolor{lpink}{rgb}{0.90, 0.70, 1.00}
\definecolor{lgreen}{rgb}{0.80, 0.95, 0.75}
\definecolor{lred}{rgb}{0.99, 0.50, 0.55}
\definecolor{lyellow}{rgb}{1.00, 0.95, 0.75}
\definecolor{llgrey}{rgb}{0.95, 0.95, 0.95}
\definecolor{salmon}{rgb}{0.99, 0.90, 0.90}

\newcommand{\Abort}{\mathsf{Abort}}

\newcommand{\Z}{\mathbb{Z}} % the integers
\newcommand{\Sim}{\mathcal{S}}

\newcommand{\etal}{\textit{et al}.}
\newcommand{\Adv}{\mathcal{A}}

\newcommand{\Func}{\mathcal{F}}
\newcommand{\Funct}[1]{\Func_\mathsf{#1}}

\newcommand{\round}[1]{\lfloor #1 \rceil}

\newenvironment{boxfig}[2]{% {#1}{#2} = {Caption}{label}
\begin{figure}[h]
  \newcommand{\FigCaption}{#1}
  \newcommand{\FigLabel}{#2}
  \vspace{-\medskipamount}
  \begin{center}
    \begin{small}
      \begin{tabular}{@{}@{~~}l@{~~}@{}}
        \begin{minipage}[b]{.94\columnwidth}
          \vspace{1ex}
          \smallskip
          }{%
        \end{minipage}\\
      \end{tabular}
    \end{small}
    \caption{\small \FigCaption}
    \label{\FigLabel}
  \end{center}
\end{figure}
}
\newenvironment{Boxfig}[3]{% {Caption}{Label}{Title}
  \begin{boxfig}{#1}{#2}
    \begin{center}
      \textbf{#3}
    \end{center}
  }{%
    \end{boxfig}
  }

\cclogo{\includegraphics{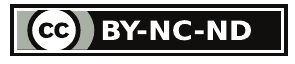}}
  
\begin{document}

  \author*[1]{Sameer Wagh}
  \author[2]{Shruti Tople}
  \author[3]{Fabrice Benhamouda}
  \author[4]{Eyal Kushilevitz}
  \author[5]{Prateek Mittal}
  \author[6]{Tal Rabin}
  \affil[1]{Princeton University~\& UC Berkeley, E-mail: swagh@princeton.edu}
  \affil[2]{Microsoft Research, E-mail: shruti.tople@microsoft.com}
  \affil[3]{Algorand Foundation, E-mail: fabrice.benhamouda@normalesup.org}
  \affil[4]{Technion, E-mail: eyalk@cs.technion.ac.il}	 
  \affil[5]{Princeton University, E-mail: pmittal@princeton.edu}
  \affil[6]{Algorand Foundation, E-mail: talrny@yahoo.com}	 

  \title{\huge \ThisWork{}: Honest-Majority Maliciously Secure Framework for Private Deep Learning}
  \runningtitle{\ThisWork{}: Honest-Majority Maliciously Secure Framework for Private Deep Learning}
  %\subtitle{...}

\begin{abstract}
{
\ch{We propose \ThisWork{}, an end-to-end 3-party protocol for efficient private training and inference of large machine learning models. \ThisWork{} presents four main advantages -- (i) It is highly {\em expressive} with support for high capacity networks  such as VGG16 (ii) it supports batch normalization which is important for training complex networks such as AlexNet (iii) \ThisWork{} guarantees {\em security with abort} against malicious adversaries, assuming an honest majority (iv) Lastly, \ThisWork{} presents new theoretical insights for protocol design that make it {\em highly efficient} and allow it to outperform existing secure deep learning solutions.} Compared to prior art for private inference, we are about $8\times$ faster than SecureNN (PETS'19) on average and comparable to ABY$^3$ (CCS'18). We are about $16-200\times$ more communication efficient than either of these. For private training, we are about $6\times$ faster than SecureNN, $4.4\times$ faster than ABY$^3$ and about $2-60\times$ more communication efficient. Our experiments in the WAN setting show that over large networks and datasets, \emph{compute operations} dominate the overall latency of MPC, as opposed to the communication.}
\end{abstract}
%  \classification[PACS]{}
 % \communicated{...}
 % \dedication{...}

  \journalname{Proceedings on Privacy Enhancing Technologies}
\DOI{Editor to enter DOI}
  \startpage{1}
  \received{..}
  \revised{..}
  \accepted{..}

  \journalyear{..}
  \journalvolume{..}
  \journalissue{..}

\maketitle

\vspace{-1.6cm}

\section{Introduction}\label{sec:intro}
\vspace{-0.2cm}
With today's digital infrastructure, tremendous amounts of private data is continuously being generated -- data which combined with deep learning algorithms can transform the current social and technological landscape. For example, distribution of child exploitative imagery  has plagued social media platforms~\cite{cei, nyt2019}. However, stringent government regulations hamper automated detection of such harmful content. Support for secure computation of state-of-the-art image classification networks would aid in detecting child exploitative imagery  on social media. Similarly, there is promise in analyzing medical data across different hospitals especially for the treatment of rare diseases~\cite{cho2018secure}. 
In both these scenarios, multiple parties (i.e., social media platforms or hospitals) could co-operate to train efficient models that have high prediction accuracy. However, the sensitive nature of such data demands deep learning frameworks that allow training on data aggregated from multiple entities while ensuring strong privacy and confidentiality guarantees.  A synergistic combination of secure computing primitives with deep learning algorithms would enable sensitive applications to benefit from the high prediction accuracies of neural networks. 

\begin{table*}[h]
\centering
\resizebox{\textwidth}{!}{
\begin{tabular}{c c c c c c c c c c c c c c  @{\hskip 0.2in}       c c c c c c c c c c c c c}
&	& \rot{Inference} &\rot{Training} & \rot{Semi-honest} & \rot{Malicious} & \rot{Linear} & \rot{Convolution} & \rot{Exact ReLU} & \rot{Maxpool} & \rot{Batch-Norm} & \rot{HE} & \rot{GC} & \rot{SS} & \rot{LAN} & \rot{WAN} & \rot{MNIST} & \rot{CIFAR-10} & \rot{Tiny ImageNet} & \rot{From~\cite{secureml}} & \rot{From~\cite{chameleon}} & \rot{From~\cite{minionn}} & \rot{LeNet} & \rot{AlexNet} & \rot{VGG-16} \\ \cmidrule(lr){3-4} \cmidrule(lr){5-6} \cmidrule(lr){7-11} \cmidrule(lr){12-14} \cmidrule(lr){15-16} \cmidrule(lr){17-19} \cmidrule(lr){20-25}
& \mr{2}{Framework} 	& \mc{2}{Private} & \mc{2}{Threat} 	& \mc{5}{Supported} 		& \mc{3}{Techniques} & \mc{2}{LAN/}		& \mc{3}{Evaluation} 	& \mc{6}{Network} \\ 
&  								& \mc{2}{Capability} & \mc{2}{Model} 		& \mc{5}{Layers} 			& \mc{3}{Used} 			&  	\mc{2}{WAN}	& \mc{3}{Dataset} 		& \mc{6}{Architectures} \\ \cmidrule(lr){3-14} \cmidrule(lr){15-25}
& & \mc{12}{Theoretical Metrics} & \mc{11}{Evaluation Metrics} \\ \midrule

\mr{7}{2PC} 
&MiniONN~\cite{minionn} &\cmark &\xmark &\cmark &\xmark &\cmark &\cmark &\cmark &\cmark &\xmark &\cmark &\cmark &\cmark &\cmark &\xmark &\cmark &\cmark &\xmark &\cmark &\xmark &\cmark &\xmark &\xmark &\xmark \\ 
&Chameleon~\cite{chameleon} &\cmark &\xmark &\cmark &\xmark &\cmark &\cmark &\cmark &\cmark &\xmark &\xmark &\cmark &\cmark          &\cmark &\cmark &\cmark &\cmark &\xmark &\cmark &\cmark &\cmark &\xmark &\smark &\xmark \\ 
&EzPC~\cite{ezpc} &\cmark &\xmark &\cmark &\xmark &\cmark &\cmark &\cmark &\cmark &\xmark &\xmark &\cmark &\cmark           &\cmark &\cmark &\cmark &\cmark &\xmark &\cmark &\xmark &\cmark &\xmark &\xmark &\xmark\\ 
&Gazelle~\cite{gazelle} &\cmark &\xmark &\cmark &\xmark &\cmark &\cmark &\cmark &\cmark &\xmark &\cmark &\cmark &\cmark       &\cmark &\xmark &\cmark &\cmark &\xmark &\cmark &\cmark &\cmark &\xmark &\xmark &\xmark  \\ 
&SecureML~\cite{secureml} &\cmark &\cmark &\cmark &\xmark &\cmark &\cmark &\cmark &\cmark &\xmark &\cmark &\cmark &\cmark       &\cmark &\smark &\cmark &\xmark &\xmark &\cmark &\xmark &\xmark &\xmark &\xmark &\xmark\\ 
&XONN~\cite{xonn} &\cmark &\xmark &\cmark &\cmark &\cmark &\cmark &\cmark &\cmark &\smark &\xmark &\cmark &\cmark        &\cmark &\xmark &\cmark &\cmark &\xmark &\cmark &\cmark &\cmark &\xmark &\xmark &\smark \\ 
&Delphi~\cite{delphi} &\cmark &\xmark &\cmark &\xmark &\cmark &\cmark &\cmark &\cmark &\xmark &\cmark &\cmark &\cmark        &\cmark &\xmark &\xmark &\cmark &\xmark &\xmark &\xmark &\cmark &\xmark &\xmark &\smark  \\ \cmidrule(lr){1-25}

\mr{6}{3PC} 
&ABY$^3$~\cite{aby3} &\cmark &\cmark &\cmark &\cmark &\cmark &\cmark &\cmark &\cmark &\xmark &\xmark &\cmark &\cmark         &\cmark &\smark &\cmark &\xmark &\xmark &\cmark &\cmark &\cmark &\xmark &\xmark &\xmark\\  
&SecureNN~\cite{securenn} &\cmark &\cmark &\cmark &\smark &\cmark &\cmark &\cmark &\cmark &\smark &\xmark &\xmark &\cmark        &\cmark &\cmark &\cmark &\xmark &\xmark &\cmark &\cmark &\cmark &\xmark &\xmark &\xmark\\ 
&CryptFlow~\cite{cryptflow} &\cmark &\xmark &\cmark &\cmark &\cmark &\cmark &\cmark &\cmark &\xmark &\xmark & \xmark  &\cmark          &\smark &\xmark &\cmark &\cmark &\cmark &\cmark &\cmark &\cmark &\cmark &\xmark &\smark\\ 
&QuantizedNN~\cite{quantizednn}* &\cmark &\xmark &\cmark &\cmark &\cmark &\cmark &\cmark &\cmark &\smark &\smark & \xmark  &\cmark          &\cmark &\cmark &\xmark &\xmark &\cmark &\xmark &\xmark &\xmark &\xmark &\xmark &\smark\\ 
&ASTRA~\cite{astra}  &\cmark &\cmark &\cmark &\smark &\cmark &\cmark &\cmark &\cmark &\xmark &\xmark &\cmark &\cmark        &\cmark &\cmark &\cmark &\xmark &\xmark &\cmark &\xmark &\xmark &\xmark &\xmark &\xmark\\ 
&BLAZE~\cite{blaze}  &\cmark &\cmark &\cmark &\cmark &\cmark &\cmark &\cmark &\cmark &\xmark &\xmark &\cmark &\cmark        &\cmark &\cmark &\smark &\xmark &\xmark &\cmark &\xmark &\xmark &\xmark &\xmark &\xmark\\ 
&Falcon (This Work) &\cmark &\cmark &\cmark &\cmark &\cmark &\cmark &\cmark &\cmark &\cmark &\xmark &\xmark &\cmark        &\cmark &\cmark &\cmark &\cmark &\smark &\cmark &\cmark &\cmark &\cmark &\cmark &\cmark \\ \cmidrule(lr){1-25}

\mr{2}{4PC} 
&FLASH~\cite{flash}  &\cmark &\cmark &\cmark &\cmark &\cmark &\cmark &\cmark &\cmark &\xmark &\xmark &\xmark &\cmark        &\cmark &\cmark &\cmark &\smark &\xmark &\cmark &\smark &\xmark &\xmark &\xmark &\xmark\\ 
&Trident~\cite{trident}  &\cmark &\cmark &\cmark &\cmark &\cmark &\cmark &\cmark &\cmark &\xmark &\xmark &\cmark &\cmark        &\cmark &\cmark &\cmark &\xmark &\xmark &\cmark &\xmark &\xmark &\xmark &\xmark &\xmark\\ 
\end{tabular}}
\caption{Comparison of various private deep learning frameworks. \falcon{} proposes efficient protocols for non-linear functionalities such as ReLU and batch normalization (1) purely using modular arithmetic (2) under malicious corruptions (3) supporting both private training and inference. \falcon{} also provides a comprehensive evaluation (1) over larger networks and datasets (2) extensively compares with related work (3) provides newer insights for future directions of PPML. \cmark{} indicates the framework supports a feature, \xmark{} indicates not supported feature, and \smark{} refers to fair comparison difficult due to the following reasons: SecureNN provides malicious privacy but not correctness and supports division but not batch norm, XONN supports a simplified batch norm specific to a binary activation layer, ABY$^3$ does not present WAN results for neural networks, Chameleon evaluates over a network similar to AlexNet but using the simpler mean-pooling operations, and due to the high round complexity and communication, SecureML provides an estimate of their WAN evaluation, Delphi evaluates over network such as ResNet-32, CryptFlow evaluates networks such as DenseNet-121, ResNet-50, uses weaker network parameters in LAN and uses ImageNet dataset. QuantizedNN uses inherent quantization of the underlying NN and performs extensive evaluation over MobileNet architectures and * refers to 3PC version among the 8 protocols. BLAZE uses a Parkinson disease dataset, similar in dimension to MNIST. FLASH and Trident use few other smaller datasets in their evaluation as well as evaluate over increasing number of layers \chh{over the network architecture from~\cite{secureml}.}}
\vspace{-0.8cm}
\label{tab:comparison}
\end{table*}

Secure multi-party computation (MPC) techniques provide a transformative capability to perform secure analytics over such data~\cite{GC,SS,secureml}. MPC provides a cryptographically secure framework for computations where the collaborating parties do not reveal their secret data to each other. The parties only learn the output of the computation on the combined data while revealing nothing about the individual secret inputs. Recently, there has been research in reducing the performance overhead of MPC protocols, specifically tailored for machine learning~\cite{securenn, aby3, xonn, gazelle, chameleon}. \ch{In this work, we focus on advancing the research in other dimensions such as expressiveness, scalability to millions of parameters, and stronger security guarantees (where parties can arbitrarily deviate from the protocol) that are necessary for practical deployment of secure deep learning frameworks.} We present \falcon --- an efficient and expressive 3-party deep learning framework that provides support for both training and inference with malicious security guarantees. Table~\ref{tab:comparison} provides a detailed comparison of \falcon with prior work.

\textbf{Contributions.}
Our contributions are as follows:

{\em (1) Malicious Security:}
\falcon provides strong security guarantees in an honest-majority adversarial setting. This assumption is similar to prior work where majority of the parties (e.g., 2 out of three) behave honestly~\cite{aby3,FLNW17}. \ThisWork{} proposes new protocols that are secure against such corruptions and ensure that either the computation always correctly completes or aborts detecting malicious activity. We achieve this by designing new protocols for the computation of non-linear functions (like ReLU).  While MPC protocols are very efficient at computing linear functions, computing non-linear functions like ReLU is much more challenging. We propose solutions both for the malicious security model and provide even more efficient protocols where semi-honest security is sufficient. We formally prove the security of \falcon using the standard simulation paradigm (see Section~\ref{sec:proofs}). We implement both the semi-honest and malicious protocols in our end-to-end framework. In this manner, \ThisWork{} provides a choice to the developers to select between either of the security guarantees depending on the trust assumption among the parties and performance requirements (improved performance for semi-honest protocols).

{\em (2) Improved Protocols:}
\ThisWork{} combines techniques from SecureNN~\cite{securenn} and ABY$^3$~\cite{aby3} that result in improved protocol efficiency. We improve the theoretical complexity of the central building block -- derivative of ReLU -- by a factor of $2\times$ through simplified algebra for fixed point arithmetic. We demonstrate our protocols in a smaller ring size, which is possible using an exact yet expensive truncation algorithm. However, this enables the entire framework to use a smaller datatype, thus reducing their communication complexity at least $2\times$. This reduced communication is critical to the communication improvements of \falcon over prior work. Furthermore, as can be seen in Section~\ref{sec:eval}, these theoretical improvements lead to even larger practical improvements due to the recursive dependence of the complex functionalities on the improved building blocks. Overall, we demonstrate how to achieve maliciously secure protocols for non-linear operations  using arithmetic secret sharing and avoiding the use of interconversion protocols (between arithmetic, Boolean and garbled circuits). 

{\em (3) Expressiveness:}
Our focus is to provide simple yet efficient protocols for the fundamental functionalities commonly used in state-of-the-art neural networks. \ch{Batch normalization, has been previously considered in privacy-preserving inference as linear transformation using Homomorphic Encryption~\cite{chabanne2017privacy, ibarrondo2018fhe, chou2018faster}. }
However, batch normalization is critical for stable convergence of networks as well as to reduce the parameter tuning required during training of neural networks. \chh{\falcon demonstrates full support for \texttt{Batch-Normalization} layers, both forward and backward pass, in private machine learning.} In other words, \falcon supports both private training and private inference. This extensive support makes \falcon expressive, thereby supporting evaluation of large networks with hundreds of millions parameters such as VGG16~\cite{vgg16} and AlexNet~\cite{alexnet} over datasets such as MNIST~\cite{mnist}, CIFAR-10~\cite{cifar} as well as Tiny ImageNet~\cite{tinyimagenet} including in both the LAN and WAN network settings. Designing secure protocols for training is more difficult due to the operations involved in back-propagation which are not required for inference. A number of prior works assume training in a trusted environment and hence provide support for only inference service~\cite{xonn,minionn,chameleon, ezpc, gazelle}. However, sensitive data is often inaccessible even during training as described in our motivating application in Section~\ref{sec:overview}.

\textbf{End-to-end Implementation and Results.}
We implement both the semi-honest and malicious variants of \falcon{} in our end-to-end framework. The codebase is written in {\tt C++} in about 14.6k LOC and will be open sourced. We experimentally evaluate the performance overhead of \falcon for both private training and inference on multiple networks and datasets. We use  $6$ diverse networks ranging from simple 3-layer multi-layer perceptrons (MLP) with about $118,000$ parameters to large networks with about 16-layers having $138$ million parameters. 
%We are also the first one to demonstrate training over networks architectures such as AlexNet due to the implementation of batch-normalization protocols.
We trained these networks on MNIST~\cite{mnist}, CIFAR-10~\cite{cifar} and Tiny ImageNet~\cite{tinyimagenet} datasets as appropriate based on the network size. We note that \falcon is one of the few private ML frameworks to support training of high capacity networks such as AlexNet and VGG16 on the Tiny ImageNet dataset. We perform extensive evaluation of our framework in both the LAN and WAN setting as well as semi-honest and malicious adversarial setting. For private inference, we are $16\times$ faster than XONN~\cite{xonn}, $32 \times$ faster than Gazelle~\cite{gazelle}, $8\times$ faster than SecureNN, and comparable to ABY$^3$ on average. For private training, we are $4.4\times$ faster than ABY$^3$ and $6\times$ faster than SecureNN~\cite{securenn}. Depending on the network, our protocols can provide \textit{up to an order of magnitude} performance improvement. \falcon is up to two orders of magnitude more communication efficient than prior work for both private training and inference. Our results in the WAN setting show that compute operations dominate the overall latency for large networks in \falcon and not the communication rounds. Hence, \falcon is an optimized 3-PC framework w.r.t. the communication which is often the bottleneck in MPC.

\vspace{-0.8cm}

\section{\ThisWork{} Overview}\label{sec:overview}
\vspace{-0.2cm}
Next, we describe the application setting for \ThisWork{}, provide a motivating application, state the threat model, and an overview of our technical contributions.

\vspace{-0.6cm}
\subsection{ A 3-Party Machine Learning Service}
\vspace{-0.3cm}
We consider the following scenario. There are two types of users, the first own data on which the learning algorithm will be applied, we call them data holders. The second are users who query the system after the learning period, we call these query users. These two sets of users need not be disjoint. We design a machine learning service. This service is provided by 3 parties which we call computing servers. We assume that government regulations or other social deterrents are sufficient enforcers for non-collusion between these computing servers. The service works in two phases: the training phase where the machine learning model of interest is trained on the data of the data holders and the inference phase where the trained model can be queried by the query users. \ch{The data holders share their data in a replicated secret sharing form~\cite{AFLNO16} between the 3 computing servers.} These 3 servers utilize the shared data and privately train the network. After this stage, query users can submit queries to the system and receive answers based on the newly constructed model held in shared form by the three servers. This way, the data holders' input has complete privacy from each of the 3 servers. Moreover, the query is also submitted in shared form and thus is kept secret from the 3 servers. 

Recent advances in MPC have rendered 3PC protocols some of the most efficient protocols in the space of privacy-preserving machine learning. Though MPC is not a broadly deployed technology yet, the 3PC adversarial model has enjoyed adoption~\cite{sharemind, securenntfe, securennpysft} due to their efficiency and simplicity of protocols. Below, we describe a concrete motivating application that would benefit from such a 3-party secure machine learning service.

\vspace{-0.5cm}
\subsubsection{Motivating Application: Detection of Child Exploitative Images Online}
\vspace{-0.2cm}
In recent years, the distribution of child exploitative imagery (CEI) has proliferated with the rise of social media platforms -- from half a million reported in between 1998-2008 to around 12 million reports in 2017 and 45 million in 2018~\cite{cei, nyt2019}.
Given the severity of the problem and stringent laws around the handling of such incidents (18 U.S. Code \S2251, 2252), it is important to develop solutions that enable efficient detection and handling of such data while complying with stringent privacy regulations. Given the success of ML, especially for image classification, it is important to leverage ML in detecting CEIs (a computer vision application). {\sc Falcon's}  approach, in contrast to deployed systems such as PhotoDNA~\cite{photodna}, enables the use of ML for this use-case. However, the inability to generate a database of the original images (due to legal regulations) leads to a problem of lack of training data. \falcon provides a cryptographically secure framework for this conundrum, where the client data is split into unrecognizable parts among a number of non-colluding entities. In this way, the solution is two-fold, MPC enables the ability to accumulate good quality training data and at the same time can enable machine-learning-as-a-service (MLaaS) for the problem of CEIs. The 3 computing parties can be Facebook, Google, and Microsoft, and will in turn be the providers of such a service. A public API can be exposed to entities willing to make use of this service, very similar to the PhotoDNA portal~\cite{photodna}. Organizations (clients of this service) that deal with significant number of CEI's can send automated requests to these 3 servers and locally reconstruct the classification result using the received responses. In terms of the adversarial model, we believe that the stringent legal framework around this application is a sufficient deterrent for these large organizations to prevent collusion among the parties. Similarly, a maliciously secure adversarial model further safeguards against individual servers being compromised. In this manner, MPC can enable an end-to-end solution to automated detection of CEIs in social media with strong privacy to the underlying data.

\vspace{-0.5cm}
\subsection{Threat Model}
\vspace{-0.2cm}
%<<<<<<< HEAD
Our threat model assumes an honest majority among the three parties in the setting described above. This is a common adversarial setting considered in previous secure multi-party computation approaches~\cite{aby3, secureml, AFLNO16, FLNW17}. We consider that one of the three parties can be either semi-honest or malicious. A semi-honest adversary passively tries to learn the secret data of the other parties while a malicious adversary can arbitrarily deviate from the protocol. We assume the private keys of each of the parties are stored securely and not susceptible to leakage. We do not protect against denial of service attacks where parties refuse to cooperate. Here, \falcon simply resorts to aborting the computation. %\sameer{Is this true?}.

\textbf{Assumptions \& Scope.}
The 3 parties each have shared point-to-point communication channels and pairwise shared seeds to use AES as a PRNG to generate cryptographically secure common randomness.
We note that as the query users receive the answers to the queries in the clear \falcon does not guarantee protecting the privacy of the training data from attacks such as model inversion, membership inference, and attribute inference~\cite{shokri2017membership,fredrikson2015model,tramer2016stealing}. Defending against these attacks is an orthogonal problem and out of scope for this work. We assume that users provide consistent shares and that model poisoning attacks are out of scope.

\vspace{-0.5cm}
\subsection{Technical Contributions}\label{subsec:techcontrib}
\vspace{-0.2cm}
In this section, we summarize some of the main contributions of this work with a focus on techniques used to achieve our results and improvements.

\textbf{Hybrid Integration for Malicious Security.}
\ThisWork{} consists of a hybrid integration of ideas from SecureNN and ABY$^3$ along with newer protocol constructions for privacy-preserving deep learning. SecureNN does not provide correctness in the presence of malicious adversaries. Furthermore, the use of semi-honest parties in SecureNN makes it a significant challenge to convert those protocols to provide security against malicious corruptions. We use replicated secret sharing (such as in~\cite{AFLNO16, FLNW17, aby3}) as our building block and use the redundancy to enforce correct behaviour in our protocols. Note that changing from the 2-out-of-2 secret sharing scheme in SecureNN to a 2-out-of-3 replicated secret sharing \ch{crucially alters some of the building blocks} -- these protocols are a new contribution of this work. We work in the 3 party setting where at most one party can be corrupt.  We prove each building block secure in the Universal Composability (UC) framework. We show that our protocols are (1) perfectly secure in the stand-alone model, i.e., the distributions are identical and not just statistically close in a model where the protocol is executed only once; and (2) have straight-line black-box simulators, i.e., only assume oracle access  and do no rewind. Theorem 1.2 from Kushilevitz~\etal~\cite{kushilevitz2010information} then implies security under general concurrent composition.

\textbf{Theoretical Improvements to Protocols.}
\ThisWork{} proposes more efficient protocols for common machine learning functionalities while providing stronger security guarantees. We achieve this through a number of theoretical improvements for reducing both the computation as well as the communication. First, in \ThisWork{} 
all parties execute the same protocol in contrast to SecureNN where the protocol is asymmetric. The uniformity of the parties leads to more optimal resource utilization. Second, the protocol for derivative of ReLU, in SecureNN~\cite{securenn} first transforms the inputs using a \texttt{Share Convert} subroutine (into secret shares modulo an odd ring) and then invokes a \texttt{Compute MSB} subroutine to compute the most significant bit (MSB) which is closely related to the DReLU function.  Note that DReLU, when using fixed point encoding over a ring $\mathbb{Z}_L$ is defined as follows:
\begin{equation}\label{eq:dreluformula}
\mathrm{DReLU}(x) =
\begin{cases}
0 	& \text{if $x>L/2$}\\
1	& \text{Otherwise}
\end{cases}
\end{equation} Each of these subroutines have roughly the same overhead. In \ThisWork{}, we show an easier technique using new mathematical insights to compute DReLU which reduces the overhead by over $2\times$. Note that ReLU and DReLU, non-linear activation functions central to deep learning, are typically the expensive operations in MPC. The first two points above lead to strictly improved protocol for these. Third, \falcon uses a smaller ring size while using an exact yet expensive truncation protocol. This trade-off however allows the entire framework to operate on smaller data-types, thus reducing the communication complexity at least $2\times$. Furthermore, this communication improvement is amplified with the superlinear dependence of the overall communication on the ring size (cf Table~\ref{tab:overheads}). This reduced communication is critical to the communication improvements of \falcon over prior work. In other words, we notice strictly larger performance improvements (than the theoretical improvements) in our end-to-end deployments of benchmarked networks presented in Section~\ref{sec:eval}.

\textbf{Improved Scope of ML Algorithms.}
\chh{Prior works focus on implementing protocols for linear layers and important non-linear operations. We propose and implement an end-to-end protocol for batch normalization (both forward and backward pass).} Batch-normalization is widely used in practice for speedy training of neural networks and is critical for machine learning for two reasons. First, it speeds up training by allowing higher learning rates and prevents extreme values of activations~\cite{batchnorm}. This is an important component of the parameter tuning for neural networks as there is limited ``seeing and learning'' during private training. Second, it reduces over-fitting by providing a slight regularization effect and thus improves the stability of training~\cite{batchnorm}.  In other words, private training of neural networks without batch normalization is generally difficult and requires significant pre-training. To truly enable private deep learning, efficient protocols for batch-normalization are required. Implementing batch normalization in MPC is hard for two reasons, first computing the inverse of a number is generally difficult in MPC. Second, most approximate approaches require the inputs to be within a certain range, i.e., there is a trade-off between having an approximate function for inverse of a number over a large range and the complexity of implementing it in MPC. Through our implementation, we enable batch normalization that can allow the training of complex network architectures such as AlexNet (about 60 million parameters).

\vspace{-0.5cm}
\subsubsection{Comprehensive Evaluation}
\vspace{-0.2cm}
As shown in Table~\ref{tab:comparison}, there are a number of factors involved in comparing different MPC protocols and that none of the prior works provide a holistic solution. We also thoroughly benchmark our proposed system -- we evaluate our approach over 6 different network architectures and over 3 standard datasets (MNIST, CIFAR-10, and Tiny ImageNet). We also benchmark our system in both the LAN and WAN setting, for training as well as for inference, and in both the semi-honest and actively secure adversarial models. Finally, we provide a thorough performance comparison against prior state-of-the-art works in the space of privacy preserving machine learning (including 2PC, purely for the sake of comprehensive comparison). We believe that such a comparison, across a spectrum of deployment scenarios, is useful for  MPC practitioners.

Finally, we note that the insights and techniques developed in this work are broadly applicable. For instance, ReLU is essentially a comparison function which can thus enable a number of other applications -- private computation of decision trees, privacy-preserving searching and thresholding, and private sorting.

\vspace{-0.8cm}

\section{Protocol Constructions}\label{sec:protocols}
\vspace{-0.2cm}
We begin by describing the notation used in this paper. We then describe how basic operations are performed over the secret sharing scheme and then move on to describe our protocols in detail. 

\vspace{-0.5cm}
\subsection{Notation}
\vspace{-0.2cm}

Let $\party{1}, \party{2}, \party{3}$ be the parties. We use $P_{i+1}, P_{i-1}$ to denote the next and previous party for $P_i$ (with periodic boundary conditions). In other words, next party for $\party{3}$ is $\party{1}$ and previous party for $\party{1}$ is $\party{3}$. We use $\Share{x}{m}$ to denote 2-out-of-3 replicated secret sharing (RSS) modulo $m$ for a general modulus $m$. For any $x$ let $\Share{x}{m} = (x_1, x_2, x_3)$ denote the RSS of a secret $x$ modulo $m$ i.e., $x\equiv x_1 + x_2 +x_3 \pmod{m}$, but they are otherwise random. We use the notation $\Share{x}{m}$ to mean $(x_1, x_2)$ is held by $\party{1}$, $(x_2, x_3)$ by $\party{2}$, and $(x_3, x_1)$ by $\party{3}$. We denote by $x[i]$ the $i$\tinyth{} component of a vector $x$. In this work, we focus on three different moduli $\ring = 2^{\ell},$ a small prime $\prm$, and $2$. In particular, we use $\ell = 2^5$, $\prm = 37$. We use fixed-point encoding with 13 bits of precision. In \protx{Mult} over $\mathbb{Z}_\prm$, the multiplications are performed using the same procedure with no truncation. ReLU, which compares a value with 0 in this representation corresponds to a comparison with $2^{\ell-1}$.
\vspace{-0.5cm}
\subsection{Basic Operations}\label{subsec:basicops}
\vspace{-0.2cm}
To ease the exposition of the protocols, we first describe how basic operations can be performed over the above secret sharing scheme. These operations are extensions of Boolean computations from Araki~\etal~\cite{AFLNO16} to arithmetic shares, similar to ABY$^3$~\cite{aby3}. However, ABY$^3$ relies on efficient garbled circuits for non-linear function computation which is fundamentally different than the philosophy of this work which relies on simple modular arithmetic. In this manner, we propose a hybrid integration of ideas from SecureNN and ABY$^3$.

\textbf{Correlated Randomness: }Throughout this work, we will need two basic random number generators. Both of these can be efficiently implemented (using local computation) using PRFs. We describe them below:

\noindent \textit{$\bullet$ 3-out-of-3 randomness:} Random $\alpha_1, \alpha_2, \alpha_3$ such that $\alpha_1 + \alpha_2 + \alpha_3 \equiv 0 \pmod{\ring}$ and party $P_i$ holds $\alpha_i$.

\noindent \textit{$\bullet$ 2-out-of-3 randomness:} Random $\alpha_1, \alpha_2, \alpha_3$ such that $\alpha_1 + \alpha_2 + \alpha_3 \equiv 0 \pmod{\ring}$ and party $P_i$ holds ($\alpha_{i}, \alpha_{i+1}$).

Given pairwise shared random keys $k_i$ (shared between parties $P_i$ and $P_{i+1}$), the above two can be computed as $\alpha_i = F_{k_i}(\mathsf{cnt}) - F_{k_{i-1}}(\mathsf{cnt})$ and $(\alpha_i, \alpha_{i-1}) = (F_{k_i}(\mathsf{cnt}),  F_{k_{i-1}}(\mathsf{cnt}))$ where $\mathsf{cnt}$ is a counter incremented after each invocation. This is more formally described later on in \protx{Prep} in Fig.~\ref{fig:prot-prep}.

\textbf{Linear operations: }Let $a,b,c$ be public constants and $\Share{x}{m}$ and $\Share{y}{m}$ be secret shared. Then $\Share{ax+by+c}{m}$ can be locally computed as $(ax_1+by_1+c, ax_2+by_2, ax_3 + by_3)$ and hence are simply local computations.

\textbf{Multiplications \protx{Mult}: }To multiply two shared values together $\Share{x}{m} = (x_1, x_2, x_3)$ and $\Share{y}{m} = (y_1, y_2, y_3)$, parties locally compute $z_1 = x_1y_1 + x_2y_1 + x_1y_2, z_2 = x_2y_2 + x_3y_2 + x_2y_3$ and $z_3 = x_3y_3 + x_1y_3+x_3y_1$. At the end of this, $z_1,z_2$ and $z_3$ form a 3-out-of-3 secret sharing of $\Share{z = x\cdot y}{m}$. Parties then perform \textit{resharing} where 3-out-of-3 randomness is used to generate 2-out-of-3 sharing by sending $\alpha_i + z_i$ to party $i-1$ \ch{where $\{\alpha_i\}$ form a 3-out-of-3 secret sharing of 0.} %For more details on fixed-point multiplication, semi-honest, and malicious variants of this refer to~\cite{aby3, AFLNO16}.

\textbf{Convolutions and Matrix Multiplications: }We rely on prior work to perform convolutions and matrix multiplications over secret shares. To perform matrix multiplications, we note that \protx{Mult} described above extends to incorporate matrix multiplications. To perform convolutions, we simply expand the convolutions into matrix multiplications of larger dimensions (cf Section 5.1 of~\cite{securenn}) and invoke the protocol for matrix multiplications. With fixed-point arithmetic, each multiplication protocol has to be followed by the truncation protocol (cf Fig.~\ref{fig:prot-prep}) to ensure correct fixed-point precision. For more details on fixed-point multiplication, semi-honest, and malicious variants of this refer to~\cite{aby3, FLNW17}.

\textbf{Reconstruction of $\Share{x}{m}$:} In the semi-honest setting, each party sends one ring element to the next party, i.e., $P_i$ sends share $x_i$ to $\party{i+1}$. In the malicious setting, each party sends $x_i$ to $\party{i+1}$ and $x_{i+1}$ to $\party{i-1}$ and aborts if the two received values do not agree. In either case, a single round of communication is required.

\textbf{Select Shares \protx{SS}: }We define a sub-routine \protx{SS}, which will be used a number of times in the descriptions of other functionalities. It takes as input shares of two random values $\Share{x}{\ring}$, $\Share{y}{\ring}$, and shares of a random bit $\Share{b}{2}$. The output $\Share{z}{\ring}$ is either $\Share{x}{\ring}$ or $\Share{y}{\ring}$ depending on whether $b = 0$ or $b=1$. To do this, we assume access to shares of a random bit $\Share{c}{2}$ and $\Share{c}{\ring}$ (pre-computation). Then we $\mathsf{open}$ the bit $(b\oplus c) = e$. If $e = 1$, we set $\Share{d}{\ring} = \Share{1-c}{\ring}$ otherwise set $\Share{d}{\ring} = \Share{c}{\ring}$. Finally, we compute $\Share{z}{\ring} =\Share{(y-x)\cdot d}{\ring} + \Share{x}{\ring}$ where $\Share{(y-x)\cdot d}{\ring}$ can be computed using \protx{Mult}$(y-x, d)$. 

\textbf{XOR with public bit $b$:} Given shares of a bit $\Share{x}{m}$ and a public bit $b$, we can locally compute shares of bit $\Share{y}{m} = \Share{x\oplus b}{m}$ by noting that $y = x + b - 2 b \cdot x$. Since $b$ is public, this is a linear operation and can be computed in both the semi-honest and malicious adversary models.

\textbf{Evaluating $\Share{(-1)^{\beta} \cdot x}{m}$ from $\Share{x}{m}$ and $\Share{\beta}{m}$:} We assume that $\beta \in \zo$. We first compute $\Share{1-2\beta}{m}$ and then perform the multiplication protocol described above to obtain $\Share{(1-2\beta) \cdot x}{m} = \Share{(-1)^{\beta} \cdot x}{m}$.
We split our computations into data dependent online computations and data independent offline computations. Protocols for offline computations are presented in Fig.~\ref{fig:prot-prep}. %The functionality $\Funct{Prep}$ for the pre-computation can be described using the functionalities in ABY$^3$~\cite{aby3}.

\vspace{-0.7cm}
\subsection{Private Compare}
\vspace{-0.2cm}
This function evaluates the bit \chh{$x \geq r$} where $r$ is public and the parties hold shares of bits of $x$ in $\bbZ_{\prm}$. Algorithm~\ref{algo:private-compare} describes this protocol. \ch{Note that $\beta$ is necessary for privacy as $\beta'$ reveals information about the output $(x \geq r)$ if not blinded by a random bit $\beta$. Each of the bits are independent so a single blinding bit $\beta$ is sufficient to hide computation of $(x \geq r)$ or $(r>x)$.}

%\TODO{Generating $m \in \mathbb{Z}_{\prm}^*$}
\begin{enumerate}[label=(\Alph*), itemsep=0pt, topsep=1mm]
\item Step~\ref{code:PC:beta}: $u[i]$ can be computed by first evaluating shares of $2\beta - 1$ and then computing the product of $(2\beta-1)$ and $x[i] - r[i]$. This can be done in a single round using one invocation of \protx{Mult}.
\item Steps~\ref{code:PC:XOR},\ref{code:PC:local}: These are simply local computations. For instance, $\Share{w[i]}{} = (w[i]_1,w[i]_2,w[i]_3)$ can be computed as $w[i]_j = x[i]_j + \delta_{j1} r[i] - 2 r[i] x[i]_j$ where $j \in \{1,2,3\}$ and $\delta_{ij}$ is the Kronecker delta function and is $1$ if $i=j$ and $0$ otherwise.
\item Step~\ref{code:PC:blind} can be computed in $\log_2 \ell+1$ rounds using sequential invocations of the \protx{Mult} with smaller strings (One additional round because of multiplication by random blinding factor).
\item Steps~\ref{code:PC:lasttwo1},\ref{code:PC:lasttwo2}: These are once again local computations.
\end{enumerate}
This protocol is an example of the challenges of integrating approaches based on simple modular arithmetic with malicious security. Both SecureNN and \falcon{}
aim to find if there exists an index $i$ such that $c[i] = 0$. However, the existence of a semi-honest third party makes checking this much easier in SecureNN. The two primary parties simply blind and mask their inputs and send them to the third party. This is not possible in~\ThisWork{} due to the stronger adversarial model and requires newer protocol constructions. In particular, we need to multiply all the $c[i]$'s together along with a mask in $\mathbb{Z}_p^*$ and reveal this final product to compute the answer.

\begin{algorithm}[t]
\footnotesize
\caption{Private Compare $\protx{PC}(\party{1},\party{2}, \party{3})$:}
\label{algo:private-compare}
\begin{algorithmic}[1]
\Require $\party{1}, \party{2}, \party{3}$ hold secret sharing of bits of $x$ in $\mathbb{Z}_{\prm}$.
\Ensure \ch{All parties get shares of the bit $(x \geq r) \in \mathbb{Z}_{2}$.}
\CommonR $\party{1}, \party{2}, \party{3}$ hold a public $\ell$ bit integer $r$, shares of a random bit in two rings $\Share{\beta}{2}$ and $\Share{\beta}{\prm}$ and shares of a random, secret integer $m \in \mathbb{Z}_{\prm}^*$.
%$\beta \in \mathbb{Z}_{\prm}$ and $\in \mathbb{Z}_{2}$, and a random secret shared integer $m \in \mathbb{Z}_{\prm}^*$.
%\hspace{-1.15cm}\textbf{Note:} All shares are over $\mathbb{Z}_{\prm}$. 

\vspace{0.2cm}
\For {$i = \{\ell-1, \ell - 2, \ldots , 0\}$} \label{code:PC:begin}
    \State Compute shares of $u[i] = (-1)^\beta (x[i] - r[i])$\label{code:PC:beta}
	\State Compute shares of $w[i] = x[i] \oplus r[i]$ \label{code:PC:XOR}
	\State Compute shares of $c[i] = u[i] + 1 + \sum_{k = i+1}^\ell w[k]$\label{code:PC:local}
\EndFor \label{codePC:end}
\State \ch{Compute and reveal $d := \Share{m}{p} \cdot \prod_{i=0}^{\ell-1} c[i] \pmod{\prm}$\label{code:PC:blind}}
\State \ch{Let $\beta' = 1$ if $(d \neq 0)$ and $0 $ otherwise.}\label{code:PC:lasttwo1}
\State \Return Shares of $\beta' \oplus \beta \in \mathbb{Z}_2$\label{code:PC:lasttwo2} 
\end{algorithmic}
\end{algorithm}

\vspace{-0.5cm}
\subsection{Wrap Function}\label{subsec:wrapfunction}
\vspace{-0.2cm}
Central to the computation of operations such as ReLU and DReLU is a comparison function. The wrap functions, $\wa{2}$ and $\wa{3}$ are defined below as a function of the secret shares of the parties and effectively compute the ``carry bit'' when the shares are added together as integers. Eq.~\ref{eq:drelu_insight} shows that DReLU can be easily computed using the $\wa{3}$ function. So all we require is a secure protocol for $\wa{3}$. We define two similar functions called ``wrap'' (denoted by $\wa{2}$ and $\wa{3}$). We call Eq.~\ref{eq:wrap3exact} the \textit{exact wrap} function. For the remainder of the paper, we use the $(\mathrm{mod}~2)$ reduction of the wrap function defined in Equation~\ref{eq:wrap3} and refer to it as simply the wrap function.
%One function takes three inputs and the other one takes four inputs and are formally defined as follows:
\begin{align}
\wa{2}(a_1, a_2, L) &= \begin{cases} \label{eq:wrap2}
0 &\text{if \ $a_1+a_2 < L$} \\
1 &\text{Otherwise}\\
\end{cases}
\\
\wa{3e}(a_1, a_2, a_3, L) &= \begin{cases} \label{eq:wrap3exact}
0 &\text{if $\sum_{i=1}^3 a_i < L$} \\
1 &\text{if $L \leq \sum_{i=1}^3 a_i < 2L$}\\
2 & \text{if $2L \leq \sum_{i=1}^3 a_i < 3L$}
\end{cases}
\end{align}
\vspace{-7mm}
\begin{equation} \label{eq:wrap3}
\wa{3}(a_1, a_2, a_3, L) = \wa{3e}(a_1, a_2, a_3, L)~(\mathrm{mod}~2)
\end{equation}

Next we briefly describe the connection between $\wa{3}$ computed on shares $a_1, a_2, a_3$ and the most significant bit (MSB) of the underlying secret $a$. Note that $a = a_1 + a_2 + a_3 \pmod{L}$ as $a_i$'s are shares of $a$ modulo $L$. Considering this sum as a logic circuit (for instance as a ripple carry adder), we can see that $\mathrm{MSB}(a) = \mathrm{MSB}(a_1) + \mathrm{MSB}(a_2) + \mathrm{MSB}(a_3) + c \pmod{2}$ where $c$ is the carry bit from the previous index. The key insight here is that the carry $c$ from the previous index is simply the $\wa{3}$ function computed on $a_i$'s (ignoring their MSB's) modulo $L/2$ (this is evident from Eq.~\ref{eq:wrap3exact}). And this last operation is synonymous with computing the $\wa{3}$ function on $2a_i$'s modulo $L$. We will further describe the consequences of this connection in Section~\ref{subsec:reluanddrelu} where we describe a protocol to compute the ReLU and DReLU functions. Algorithm~\ref{algo:wrap} gives the protocol for securely computing the $\wa{3}$ function. 

\begin{algorithm}[t]
\footnotesize
\caption{wrap$_3$ $\protx{WA}(\party{1},\party{2}, \party{3})$:}
\label{algo:wrap}
\begin{algorithmic}[1]
\Require $\party{1}, \party{2}, \party{3}$ hold shares of $a$ in $\mathbb{Z}_{\ring}$.
\Ensure $\party{1}, \party{2}, \party{3}$ get shares of a bit $\theta = \wa{3}(a_1, a_2, a_3, L)$
\CommonR \ch{$\party{1}, \party{2}, \party{3}$ hold shares $\Share{x}{\ring}$ (of a random number $x$), $\Share{x[i]}{\prm}$ (shares of bits of $x$) and $\Share{\alpha}{2}$ where $\alpha = \wa{3}(x_1, x_2, x_3, L)$.}

\vspace{0.2cm}
\State \label{code:wrapstart}\ch{Compute $r_j \equiv a_j + x_j \pmod{L}$ and $\beta_j = \wa{2}(a_j, x_j, L)$}
\State \ch{Reconstruct $r \equiv \sum r_j \pmod{L}$}
\State \ch{Compute $\delta = \wa{3}(r_1, r_2, r_3, L)$ \Comment{In the clear}}
\State \ch{Run $\protx{PC}$ on $x, r+1$ to get $\eta = (x \geq r+1)$.\label{code:wrapend}}
\State \Return $\theta = \beta_1 + \beta_2 + \beta_3 + \delta - \eta - \alpha$
\end{algorithmic}
\end{algorithm}

%\begin{algorithm}[t]
%\footnotesize
%\caption{wrap$_3$ $\protx{WA}(\party{1},\party{2}, \party{3})$:}
%\label{algo:wrap}
%\begin{algorithmic}[1]
%\Require $\party{1}, \party{2}, \party{3}$ hold shares of $a$ in $\mathbb{Z}_{\ring}$.
%\Ensure $\party{1}, \party{2}, \party{3}$ get shares of a bit $\theta = \wa{3}(a_1, a_2, a_3, L)$
%\CommonR $\party{1}, \party{2}, \party{3}$ hold shares $\Share{r}{\ring}$ (random number), $\Share{r[i]}{\prm}$ (shares of bits of $r$) and $\Share{\alpha}{2}$ where  $\alpha = \wa{3}(r_1, r_2, r_3, L)$.
%
%\vspace{0.2cm}
%\State \label{code:wrapstart}Compute $x_j \equiv a_j + r_j \pmod{L}$ and $\beta_j = \wa{2}(a_j, r_j, L)$
%\State Reconstruct $x \equiv \sum x_j \pmod{L}$
%\State Compute $\delta = \wa{3}(x_1, x_2, x_3, L)$ \Comment{In the clear}
%\State Run $\protx{PC}$ on $x, r$ to get $\eta = (r > x)$.\label{code:wrapend}
%\State \Return $\theta = \beta_1 + \beta_2 + \beta_3 + \delta - \eta - \alpha$
%\end{algorithmic}
%\end{algorithm}

Note that $\wa{2}$ function is always computed locally and hence a secure algorithm is not needed for the same. \ch{Furthermore, note that the $\wa{2}$ function allows us to write exact integer equations as follows: if $a \equiv a_1 + a_2 \pmod{L}$ then $a = a_1 + a_2 - \wa{2}(a_1, a_2, L) \cdot L$ where the former relation is a congruence relation but the latter is an integer relation (and has exact equality). Finally, to see the correctness of the $\wa{3}$ protocol, in reference to Algorithm~\ref{algo:wrap}, we can write the following set of equations 
\begin{align}
r &= a+x -\eta \cdot \ring \label{corr:1}\\
r & = r_1 + r_2 + r_3 - \delta_e \cdot \ring \label{corr:2}\\
r_i &= a_i + x_i - \beta_i \cdot \ring ~~~~~ \forall i \in \{1,2,3\} \label{corr:3}\\
x &= x_1 + x_2 + x_3 - \alpha_e \cdot \ring \label{corr:4}
\end{align}}%
\noindent where $\delta_e, \alpha_e$ denote the exact wrap functions, Eq.~\ref{corr:2},\ref{corr:4} follow from the definition of the exact wrap function, while Eq~\ref{corr:3} \chh{follows from the definition of $\wa{2}$ function. To see Eq.~\ref{corr:1}, note that $r, a, x \in [0, L-1]$ and that $r\equiv a + x \pmod{L}$. Hence $a + x \geq L$ iff $r < x$ (or $x \geq r+1$).} Finally, assuming $\theta_e$ is the exact wrap function on $a_1, a_2, a_3$ i.e.,  
\begin{align}
a & = a_1 + a_2 + a_3 - \theta_e \cdot \ring ~~~~~ ~~~~~~~~~\label{corr:target}
\end{align}
Eqs.~\ref{corr:1}-\ref{corr:target} together give a constraint among the Greek symbols (in other words, (\ref{corr:1}) - (\ref{corr:2}) - (\ref{corr:3}) + (\ref{corr:4}) + (\ref{corr:target}) gives Eq.~\ref{eq:mainwrap} below)
\begin{equation}\label{eq:mainwrap}
    \theta_e = \beta_1 + \beta_2 + \beta_3 + \delta_e - \eta - \alpha_e ~~~~~~~
\end{equation}
Reducing Eq.~\ref{eq:mainwrap} modulo 2 gives us $\theta = \beta_1 + \beta_2 + \beta_3 + \delta - \eta - \alpha$ which is used to compute $\wa{3}$ as in Algorithm~\ref{algo:wrap}.

\vspace{-0.5cm}
\subsection{ReLU and Derivative of ReLU}\label{subsec:reluanddrelu}
\vspace{-0.2cm}
We now describe how to construct a protocol for securely computing $\mathrm{ReLU}(a)$ and $\mathrm{DReLU}(a)$ for a given secret $a$. Recall that we use fixed point arithmetic over $\bbZ_{2^{\ell}}$ for efficiency reasons. Using the natural encoding of native \texttt{C++} data-types, we know that positive numbers are the first $2^{\ell-1}$ and have their most significant bit equal to 0. Negative numbers, on the other hand are the last $2^{\ell-1}$ numbers in the $\ell$-bit range and have their most significant bit equal to 1. Thus, the DReLU function defined by Eq.~\ref{eq:dreluformula}, has a simple connection with the most significant bit (MSB) of the fixed point representation viz., $\mathrm{DReLU}(a) = 1 - \mathrm{MSB}(a)$. Furthermore, in Section~\ref{subsec:wrapfunction}, we have seen the connection between $\mathrm{MSB}(a)$ and $\wa{3}$. Together, these insights can be distilled into the following equation:
%The ReLU function on a secret $a$ is defined as the $\max (a, 0)$ while the derivative of ReLU (DReLU) is simply the indicator bit whether the secret is positive or negative. We consider fixed-point arithmetic over $\bbZ_{2^{\ell}}$ for efficiency reasons. Using the natural encoding of native \texttt{C++} data-types, we know that positive numbers are the first $2^{\ell-1}$ and negative numbers are the last $2^{\ell-1}$ numbers in the $\ell$-bit range. In other words, the MSB of the $\ell$-bit representation of positive numbers is 0 and that of negative numbers is 1. This close connection between DReLU and the MSB enables us to write the following:
\begin{align}\label{eq:drelu_insight}
   \begin{split}
		\mathsf{DReLU(a)} = \mathsf{MSB}(a_1) \oplus \mathsf{MSB}(a_2) \oplus \mathsf{MSB}(a_3) \\
		\oplus~ \wa{3}(2a_1, 2a_2, 2a_3, L) \oplus 1
   \end{split}
\end{align}
In particular, Derivative of ReLU can be computed by combining the output of the wrap function with local computations. Finally, for computing ReLU from DReLU, we simply call \protx{SS} (which effectively performs \protx{Mult} on shares of $a$ and shares of $\mathsf{DReLU}(a)$). With these observations, we can implement the ReLU and Derivative of ReLU protocols (see Algorithm~\ref{algo:relu}). Note that the approach here is crucially different from the approach SecureNN uses due to use of fundamentally different building blocks as well as deeper mathematical insights such as Eq.~\ref{eq:drelu_insight}. To achieve the DReLU functionality, SecureNN first uses a subroutine to transform the shares of the secret into an odd modulus ring and then uses another subroutine to compute the MSB (cf Section~\ref{subsec:techcontrib}). Both these subroutines have similar complexities. \ThisWork{} on the other hand uses the insight presented in Eq.~\ref{eq:drelu_insight} to completely eliminate the need for these subroutines, improving the efficiency by about $2 \times$ and simplifying the overall protocol. This also drastically improves the end-to-end performance (by over $6.4\times$) as the ReLU and DReLU functionalities are the building blocks of every comparison in the network.

\begin{algorithm}[t]
\footnotesize
\caption{ReLU, $\Pi_{\mathsf{{ReLU}}}(P_{1}, P_{2}, P_{3})$:}
\label{algo:relu}
\begin{algorithmic}[1]

\Require $\party{1}, \party{2}, \party{3}$ hold shares of $a$ in $\mathbb{Z}_{\ring}$.
\Ensure $\party{1}, \party{2}, \party{3}$ get shares of $\mathsf{ReLU}(a).$ %given by the RHS of Eq.~\ref{eq:wrap3}
% \Require Parties hold $\Share{a}{\ring}$.
% \Ensure Parties get $\Share{\mathrm{ReLU}(a)}{\ring}$. 
\CommonR $\Share{c}{2}$ and $\Share{c}{\ring}$ (shares of a random bit in two rings)
%\hspace{-1.15cm}\textbf{Common Randomness:} $\Share{c}{2}$ and $\Share{c}{\ring}$ (shares of a random bit in different rings)
%Shares of bit $c \in \mathbb{Z}_2$ and $c \in \mathbb{Z}_{\ring}$. %$\party{0}, \party{1}, \party{2}$ hold shares of random bit $\Share{c}{2}$ and $\Share{m_c}{\ring}$ where $m_i$ for $i \in \{0,1\}$ are random shares of $i \in  \mathbb{Z}_{\ring}$.

\vspace{0.2cm}
\State Run $\protx{WA}$ to get $\wa{3}(2a_1, 2a_2, 2a_3, L)$
\State Compute $\Share{b}{2}$ where $b= \mathsf{DReLU}(a)$ \Comment{Local comp. (Eq.~\ref{eq:drelu_insight})}
\State \Return Output of \protx{SS} run on $\{a, 0\}$ with $b$ as selection.
%\State Reconstruct $e = b \oplus c$
%\If {$e=1$}
%\State Set $\Share{m_b}{\ring} = \Share{1 - c}{\ring}$
%\Else 
%\State Set $\Share{m_b}{\ring} = \Share{c}{\ring}$
%\EndIf
%\State \Return Output of \protx{Mult} on inputs $(a, m_b)$
\end{algorithmic}

\end{algorithm}
% 	\If{$f = 1$}
% 		\State \Return \protx{Mult}$(a, m_c)$
% 	\Else
% 		\State \Return \protx{Mult}$(a, 1 - m_c)$ 
% 	\EndIf

\vspace{-0.5cm}
\subsection{Maxpool and Derivative of Maxpool}
\vspace{-0.2cm}
The functionality of maxpool simply takes as input a vector of secret shared values and outputs the maximum value. For derivative of maxpool, we need a one-hot vector of the same size as the input where the 1 is at the location of the index of the maximum value. Maxpool can be implemented using a binary sort on the vector of inputs and small amounts of bookkeeping, where the comparisons can be performed using ReLUs. Derivative of maxpool can be efficiently implemented along with maxpool described in Algorithm~\ref{algo:maxpool}. %Note that Steps~\ref{algo:sketchy1}-\ref{algo:sketchy2} have to implemented privately and can be done using \protx{SS} as described in Section~\ref{subsec:basicops}.
\begin{algorithm}
\footnotesize
\caption{Maxpool, $\Pi_{\mathsf{{Maxpool}}}(P_{1}, P_{2}, P_{3})$:}
\label{algo:maxpool}
\begin{algorithmic}[1]

\Require $\party{1}, \party{2}, \party{3}$ hold shares of $a_1, a_2, \hdots a_n$ in $\mathbb{Z}_{\ring}$.
\Ensure $\party{1}, \party{2}, \party{3}$ get shares of $a_k$ and $\bm{e_k}$ where $k = \mathsf{argmax} \{a_1, \hdots a_n\}$ and where $\bm{e_k} = \{e_1, e_2, \hdots e_n\}$ with $e_i = 0 ~ \forall i \neq k$ and $e_k = 1$. 
\CommonR No additional common randomness required.

\vspace{0.2cm}
\State Set $\mathsf{max} \leftarrow a_1$ and $\bm{\mathsf{ind}} \leftarrow \bm{e_1} = \{1,0,\hdots, 0\}$
\For {$i = \{2, 3, \hdots n\}$} 
	\State Set $\mathsf{d_{max}} \leftarrow (\mathsf{max} - a_i)$ and $\bm{\mathsf{d_{ind}}} \leftarrow (\bm{\mathsf{ind} - e_i})$
	\State $b \leftarrow \protx{\mathsf{DReLU}}(\mathsf{d_{max}})$ \Comment{b $\rightarrow$ Derivative of ReLU}
	\State $\mathsf{max} \leftarrow$ \protx{SS} on inputs $\{a_i,  \mathsf{max}, b\}$.
	\State $\bm{\mathsf{ind}} \leftarrow$ \protx{SS} on inputs $\{\bm{e_i}, \bm{\mathsf{ind}}, b\}$.
	%\State Set $\mathsf{max}$ as \protx{SS} output on inputs $\{a_i,  \mathsf{d_{max}} + a_i\}$ using selection $b$.
	%\State Set $\bm{\mathsf{ind}}$ as \protx{SS} output on inputs $\{\bm{e_i}, \bm{\mathsf{d_{ind}}} + \bm{e_i}\}$ using selection $b$.
%	\If {b = 0} \label{algo:sketchy1}
%		\State $\mathsf{max} ~= a_i$ and $\bm{\mathsf{ind}} ~= \bm{e_i}$ 
%	\Else
%		\State $\mathsf{max} ~= \mathsf{d_{max}} + a_i$ and $\bm{\mathsf{ind}} ~= \bm{\mathsf{d_{ind}}} + \bm{e_i}$ 
%	\EndIf \label{algo:sketchy2}
%	\State $\mathsf{max} ~=$ \protx{SS} $(\{0,\mathsf{d_{max}}\}, b) + a_i$  
%	\State $\bm{\mathsf{ind}} ~=$ \protx{SS} $(\{0,\bm{\mathsf{d_{ind}}}\}, b) + \bm{e_i}$ 
\EndFor 
\State \Return $\mathsf{max}, \bm{\mathsf{ind}}$
\end{algorithmic}
\end{algorithm}

\vspace{-0.5cm}
\subsection{Division and Batch Normalization}
\vspace{-0.2cm}
Truncation allows parties to securely eliminate lower bits of a secret shared value (i.e., truncation by $k$ bits of a secret $a \rightarrow a/2^k$). However, the problem of dividing by a secret shared number is considerably harder and efficient algorithms rely on either (1) sequential comparison (2) numerical methods. In this work, we use the numerical methods approach for its efficiency. We use the specific choices of initializations given in~\cite{catrina2010secure, aliasgari2013secure} to efficiently compute division over secret shares. A crucial component of numerical methods is the need to estimate the value of the secret within a range. We achieve this using Algorithm~\ref{algo:power}.
Note that Algorithm~\ref{algo:power} outputs the bounding power of 2, which is also what is guaranteed by the functionality. In this way, we only reveal the bounding power of 2 and nothing else.
\begin{algorithm}[t]
\footnotesize
\caption{Bounding Power, $\Pi_{\mathsf{{Pow}}}(P_{1}, P_{2}, P_{3})$:}
\label{algo:power}
\begin{algorithmic}[1]

\Require \ch{$\party{1}, \party{2}, \party{3}$ hold shares of $x$ in $\mathbb{Z}_{\ring}$.}
\Ensure \chh{$\party{1}, \party{2}, \party{3}$ get $\alpha$ in the clear, where $2^{\alpha} \leq x < 2^{\alpha+1}$.}
\CommonR No additional common randomness required.

\vspace{0.2cm}
\State \ch{Initialize $\alpha \leftarrow 0$ \label{code:Div1}}
\For {\chh{$i = \{\ell-1, \hdots , 1, 0\}$}}
%	\State \ch{Set $d_y \leftarrow (y - 2^{2^i + \alpha})$}
	\State \ch{$c \leftarrow \protx{\mathsf{DReLU}}(\mathsf{x - 2^{2^i + \alpha}})$ and reconstruct $c$}
	\State \chh{Set $\alpha \leftarrow \alpha + 2^i$ if $c=1$}
%	\If {$c=1$}
%	\State \chh{$\alpha \leftarrow \alpha + 2^i$}
%	\EndIf
\EndFor \label{code:Div2}
\State \Return $\alpha$
\end{algorithmic}
\end{algorithm}

%\begin{algorithm}[t]
%\footnotesize
%\caption{Power, $\Pi_{\mathsf{{Pow}}}(P_{1}, P_{2}, P_{3})$:}
%\label{algo:power}
%\begin{algorithmic}[1]
%
%\Require \ch{$\party{1}, \party{2}, \party{3}$ hold shares of $x$ in $\mathbb{Z}_{\ring}$.}
%\Ensure \chh{$\party{1}, \party{2}, \party{3}$ get $\alpha$ in the clear, where $2^{\alpha} \leq x < 2^{\alpha+1}$.}
%\CommonR No additional common randomness required.
%
%\vspace{0.2cm}
%\State \ch{Set $y \leftarrow x$ and $\alpha \leftarrow 0$ \label{code:Div1}}
%\For {\chh{$i = \{\ell-1, \hdots , 2, 1, 0\}$}}
%	\State \ch{Set $d_y \leftarrow (y - 2^{2^i + \alpha})$}
%	\State \ch{$c \leftarrow \protx{\mathsf{DReLU}}(d_y)$ and reconstruct $c$}
%	\If {$c=1$}
%	\State \ch{Set $y \leftarrow d_y$ and $\alpha \leftarrow \alpha + 2^i$}
%	\EndIf
%\EndFor \label{code:Div2}
%\State \Return $\alpha$
%\end{algorithmic}
%\end{algorithm}

Algorithm~\ref{algo:division} is used to compute the value of $a/b$ where $a,b$ are secret shared. The first step for the algorithm is to transform $b \rightarrow x$ where $x \in [0.5, 1)$. Note that even though $b$ is a fixed-point precision of $f_p$, \chh{for the computations in Algorithm~\ref{algo:division}, $x$ has to be interpreted as a value with fixed-point precision $\alpha + 1$ where $2^{\alpha} \leq b < 2^{\alpha+1}$}. Thus we first need to extract $\alpha$ (the appropriate range) using Algorithm~\ref{algo:power}. Let $w_0 = 2.9142 - 2x$, $\epsilon_0 = 1 - x \cdot w_0$ (cf.~\cite{catrina2010secure, aliasgari2013secure} for choice of constants). Then an initial approximation for $1/x$ is $w_0 \cdot (1+\epsilon_0)$. For higher order approximations, set $\epsilon_i = \epsilon_{i-1}^2$ and multiply the previous approximate result by $(1+\epsilon_i)$ to get a better approximate result. Each successive iteration increases the round complexity by 2. For our value of fixed-point precision, we use the following approximation which works with high accuracy (refer to Section~\ref{sec:eval} for details):\vspace{-4mm}
\begin{algorithm}
\footnotesize
\caption{Division, $\Pi_{\mathsf{{Div}}}(P_{1}, P_{2}, P_{3})$:}
\label{algo:division}
\begin{algorithmic}[1]

\Require $\party{1}, \party{2}, \party{3}$ hold shares of $a, b$ in $\mathbb{Z}_{\ring}$.
\Ensure $\party{1}, \party{2}, \party{3}$ get shares of $a/b$ in $\mathbb{Z}_{\ring}$ computed as integer division with a given fixed precision $f_p$.
\CommonR No additional common randomness required.
%\hspace{-1.15cm}\textbf{Common Randomness:} No additional common randomness required.

\vspace{0.2cm}
% \State Set $\mathsf{x} \leftarrow b$ and $\alpha \leftarrow 0$ \label{code:Div1}
% \For {$i = \{\ell-1, \hdots , 2, 1\}$} 
% 	\State Set $d_x \leftarrow (x - 2^{2^i + \alpha})$
% 	\State $c \leftarrow \protx{\mathsf{DReLU}}(\mathsf{d_{x}})$ and reconstruct $c$
% 	\State If $c = 1$ set $x \leftarrow d_x$ and $\alpha \leftarrow \alpha + 2^i$
% \EndFor \label{code:Div2}
\State Run \protx{Pow} on $b$ to get $\alpha$ such that \chh{$2^{\alpha} \leq b < 2^{\alpha+1}$}
\State Compute $w_0 \leftarrow 2.9142 - 2b$
\State Compute $\epsilon_0 \leftarrow 1 - b \cdot w_0$ and $\epsilon_1 \leftarrow \epsilon_0^2$ %and $\epsilon_2 \leftarrow \epsilon_1^2$
\State \Return $a w_0 (1 + \epsilon_0)(1+\epsilon_1)$%(1+\epsilon_2)$
\end{algorithmic}
\end{algorithm}
\vspace{-8mm}

\begin{algorithm}[h]
\footnotesize
\caption{Batch Norm, $\Pi_{\mathsf{{BN}}}(P_{1}, P_{2}, P_{3})$:}
\label{algo:batchnorm}
\begin{algorithmic}[1]

\Require $\party{1}, \party{2}, \party{3}$ hold shares of $a_1, a_2 \hdots a_m$ in $\mathbb{Z}_{\ring}$ where $m$ is the size of each batch and shares of two learnable parameters $\gamma, \beta$. 
\Ensure $\party{1}, \party{2}, \party{3}$ get shares of $\gamma z_i + \beta$ for $i \in [m]$ and $z_i = (a_i - \mu)/(\sqrt{\sigma^2 + \epsilon})$ where $\mu = 1/m \sum a_i$, $\sigma^2 = 1/m \sum (a_i - \mu)^2$, and $\epsilon$ is a set constant.
\CommonR No additional common randomness required.
%\hspace{-1.15cm}\textbf{Common Randomness:} No additional common randomness required.

\vspace{0.2cm}
\State Set $\mu \leftarrow 1/m \cdot \sum a_i$ 
\State Compute $\sigma^2 \leftarrow 1/m \cdot \sum (a_i -\mu)^2$ and let $b = \sigma^2 + \epsilon$
\State Run \protx{Pow} on $b$ to find $\alpha$ such that \chh{$2^{\alpha} \leq  b < 2^{\alpha+1}$}
\State Set $x_0 \leftarrow 2^{-\round{\alpha/2}}$
\For {$i \in 0, \hdots, 3$}
	\State Set $x_{i+1} \leftarrow \frac{x_{i}}{2}(3 - b x_{i}^2)$
\EndFor
\State \Return $\gamma \cdot x_{\mathsf{rnds}} \cdot (a_i - \mu) + \beta$ for $i \in [m]$
\end{algorithmic}
\end{algorithm}
\vspace{-6mm}

%\mathsf{AppDiv}(x) &= w_0 \cdot (1+\epsilon_0)(1+\epsilon_1)(1+\epsilon_2) \approx \frac{1}{x}
\begin{equation}\label{eq:appdiv}
\begin{aligned}
\mathsf{AppDiv}(x) &= w_0 \cdot (1+\epsilon_0)(1+\epsilon_1) \approx \frac{1}{x}
\end{aligned}
\end{equation}
Batch-norm is another important component of neural network architectures. They improve the convergence as well as help automate the training process. Algorithm~\ref{algo:batchnorm} describes the protocol to compute batch-norm.
For step 3, we use Newton's method and use $2^{-\round{\alpha/2}}$ as an initial approximation of $1/\sqrt{\sigma^2 + \epsilon}$, where \chh{$2^{\alpha} \leq \sigma^2 + \epsilon < 2^{\alpha+1}$} and use the successive iterative formula:
\begin{equation}\label{eq:invsqrt}
    x_{n+1} = \frac{x_n}{2} \left( 3- ax_n^2 \right)
\end{equation}
Given the choice of initial guess, 4 rounds are sufficient for a close approximation with our choice of fixed-point precision. However, batch normalization during training is computed by sequentially computing $\sqrt{\sigma^2+\epsilon}$ and then computing the inverse. This approach is used to optimize the computation required during back-propagation which requires the values of $\sqrt{\sigma^2+\epsilon}$. For computing the square root of a value $a$, we use Newton's method given by Eq.~\ref{eq:sqrt}. %Alternatively, we can use the iterative approach given by Eq.~\ref{eq:sqrt} and
This can then be used in conjunction with the inverse computation given by Eq.~\ref{eq:appdiv} to complete the batch-norm computations.
\begin{equation}\label{eq:sqrt}
    x_{n+1} = \frac{1}{2} \left( x_n + \frac{a}{x_n} \right)
\end{equation}

\vspace{-0.8cm}

%\vspace{-0.3cm}
\section{Theoretical Analysis}\label{sec:theoretical}
\vspace{-0.2cm}
We provide theoretical analysis of our framework and protocols. In particular, we provide proofs of security and analyze the theoretical complexity.

\vspace{-0.7cm}
\subsection{Security Proofs}\label{sec:proofs}
\vspace{-0.2cm}
We model and prove the security of our construction in the real world-ideal world simulation paradigm~\cite{gmw87,JC:Canetti00,Can01}. In the real interaction, the parties execute the protocol in the presence of an adversary and the environment. On the other hand, in the ideal interaction, the parties send their inputs to a trusted party that computes the functionality truthfully. Finally, to prove the security of our protocols, for every adversary in the real interaction, there exists a simulator in the ideal interaction such that the environment cannot distinguish between the two scenarios. In other words, whatever information the adversary extracts in the real interaction, the simulator can extract it in the ideal world as well.

We show that our protocols are perfectly secure (i.e., the joint distributions of the inputs, outputs, and the communication transcripts are exactly the same and not statistically close) in the stand-alone
model (i.e., protocol is executed only once), and that they have a straight-line black-box simulators (i.e., only assume oracle access to the adversary and hence do no rewind). We then rely on the result of Kushilevitz~\etal~\cite{kushilevitz2010information} to prove that our protocols are secure under concurrent general composition (Theorem 1.2 in~\cite{kushilevitz2010information}).

Due to space constraints, we formally describe the functionalities in Appendix~\ref{app:functionalities}. We describe simulators for \protx{PC} (Fig.~\ref{FPC}), \protx{WA} (Fig.~\ref{FWA}), \protx{ReLU} (Fig.~\ref{FRELU}), \protx{Maxpool} (Fig.~\ref{FMAX}), \protx{Pow} (Fig.~\ref{FPOW}), \protx{Div} (Fig.~\ref{FDIV}), and \protx{BN} (Fig.~\ref{FBN}) that achieve indistinguishability. $\Funct{Mult}, \Funct{Trunc}, \Funct{Reconst}$ are identical to prior works~\cite{aby3, FLNW17}.
We prove security using the standard indistinguishability argument. To prove the security of a particular functionality, we set-up hybrid interactions where the sub-protocols used in that protocol are replaced by their corresponding ideal functionalities and then prove that the interactions can be simulated. This hybrid argument in effect sets up a series of interactions $I_0, I_1, \hdots I_k$ for some $k$ where $I_0$ corresponds to the real interaction and $I_k$ corresponds to the ideal interaction. Each neighboring interaction, i.e., $I_i, I_{i+1}$ for $i \in \{0, \hdots, k-1\}$ is then shown indistinguishable from each other, in effect showing that the real and ideal interactions are indistinguishable. Without loss of generality, we assume that party $P_2$ is corrupt. In the \textit{real world}, the adversary $\Adv$ interacts with the honest parties $P_0$ and $P_1$. In the \textit{ideal world}, the simulator interacts with the adversary and simulates exact transcripts for interactions between the adversary $\Adv$ and $P_0, P_1$. On the other hand, the simulator should be able to extract the adversaries inputs (using the values for the inputs of the honest parties in the internal run and the fact that each honest party has one component of the replicated secret sharing). These inputs are fed to the functionality to generate correct output distributions, thus achieving security against malicious adversaries. Theorems~\ref{thm:PC}-\ref{thm:divbn} gives the indistinguishability of these two interactions.

\begin{theorem}\label{thm:PC}
\protx{PC} securely realizes $\Funct{PC}$ with abort, in the presence of one malicious party in the $(\Funct{Mult},\Funct{Reconst},\Funct{Prep})$-hybrid model.
\end{theorem}
\begin{proof}
We first set up some detail on the proof strategy that is essential for other proofs as well. For the ease of exposition, we describe it in the context of \protx{PC}. The goal of designing a simulator is to be able to demonstrate the ability to produce transcripts that are indistinguishable from the transcripts in the real world. The joint distribution of the inputs and outputs is a part of these transcripts and hence has to be indistinguishable in the two interactions. However, since the honest parties simply forward their inputs to the functionality, the simulator must be able to extract the inputs of the malicious parties to be able to generate the correct shares for the honest parties.

The usual technique to achieve this is to have the simulator run a simulated version of the protocol internally, i.e., emulating the roles of the honest parties and interacting with the adversary. This is what we call an \textit{internal run}. This internal run can then be used to extract the inputs of the adversarial party (which can then be forwarded to the functionality in the ideal interaction). In the hybrid argument, the subroutines used in the protocol can be replaced by their corresponding ideal interactions, the simulator can emulate the roles of these trusted functionalities in its internal run.

In the specific context of \protx{PC}, the simulator $\Sim$ for adversary $\Adv$ works by playing the role of the trusted party for $\Funct{Mult}, \Funct{Reconst}$ and $\Funct{Prep}$. To be able to simulate, we need to show that:

\vspace{2mm}
\noindent (1) Real interaction transcripts can be simulated.

\noindent (2) Honest parties receive their outputs correctly.
\vspace{2mm}

\noindent Simulation follows easily from the protocol and the hybrid argument. The simulator for \protx{Mult} (along with the simulator for \protx{Reconst}) can be used to simulate the transcripts from Steps~\ref{code:PC:beta}, \ref{code:PC:blind} (from Algorithm~\ref{algo:private-compare}). Note that the distributions of these transcripts are all uniformly random values \ch{($\beta$ is required to make the transcript for $\beta'$ uniformly random, the various bits $u[i], w[i]$, and $c[i]$ are random because $x$ is random)} and hence achieve perfect security. Steps~\ref{code:PC:XOR}, \ref{code:PC:local}, \ref{code:PC:lasttwo1}, and \ref{code:PC:lasttwo2} on the other hand are all local and do not need simulation.

To extract the inputs of the malicious party, the simulator uses the fact that it has access to $r$ and $\beta$ (though $\Funct{Prep}$) and all the internal values for the honest parties (in the internal run) and hence can extract the shares of $x[i]$ from the corrupt party $P_2$. Finally, if the protocol aborts at any time in the internal run, then the simulator sends $\Abort$ to $\Funct{PC}$ otherwise, it inputs the extracted shares of $x[i]$ to $\Funct{PC}$ and the honest parties receive their outputs.
\end{proof}

\begin{theorem}\label{thm:WA}
\protx{WA} securely realizes $\Funct{WA}$ with abort, in the presence of one malicious party in the $(\Funct{Mult},\Funct{PC},\Funct{Reconst},\Funct{Prep})$-hybrid model.
\end{theorem}
\begin{proof}
We use a similar set-up as the proof of Theorem~\ref{thm:PC}. Step 1 is local computation and does not need simulation. Steps 2, 4 can be simulated using the simulators for $\Funct{Reconst}, \Funct{PC}$ respectively. Input extraction follows from having access to $r_i$ (through $\Funct{Prep}$) and output $x$ if the protocol does not abort. If the protocol does abort at any time in the internal run, then the simulator sends $\Abort$ to $\Funct{WA}$. Otherwise, it simply passes on the extracted shares of $a[i]$ to $\Funct{WA}$ and the honest parties receive their outputs. Note that \protx{DReLU} is not formally defined. However, this is simply local computation over \protx{WA} and the proofs can be extended analogously.
\end{proof}

\begin{theorem}\label{thm:relu}
\protx{ReLU} securely realizes $\Funct{ReLU}$ with abort, in the presence of one malicious party in the $(\Funct{Mult},\Funct{WA},\Funct{Prep})$-hybrid model.
\end{theorem}
\begin{proof}
Simulation is done using the hybrid argument. The protocol simply composes $\Funct{WA}$ and $\Funct{Mult}$ and hence is simulated using the corresponding simulators.
\end{proof}

\begin{theorem}\label{thm:maxpool}
\protx{Maxpool} securely realizes $\Funct{Maxpool}$ with abort, in the presence of one malicious party in the $(\Funct{Mult},\Funct{ReLU},\Funct{Prep})$-hybrid model.
\end{theorem}
\begin{proof}
Similar to the proof of Theorem~\ref{thm:relu}, simulation works by sequentially composing the simulators for $\Funct{ReLU}$ and $\Funct{Mult}$.
\end{proof}

\begin{theorem}\label{thm:pow}
\protx{Pow} securely realizes $\Funct{Pow}$ with abort, in the presence of one malicious party in the $(\Funct{Mult},\Funct{ReLU},\Funct{Reconst},\Funct{Prep})$-hybrid model.
\end{theorem}
\begin{proof}
The simulator for $\Adv$ works by playing the role of the trusted party for $\Funct{Mult}, \Funct{ReLU},$ and $\Funct{Reconst}$. The protocol sequentially reveals bits of the scale $\alpha$. It is important to note the functionality that it emulates (see  in Fig.~\ref{FPOW}). The simulator runs the first iteration of the loop and in the process extracts the adversaries inputs. Then it proceeds to complete all the iterations of the loop. If the protocol proceeds without aborting till the end, then the simulator sends the extracted shares of $b$ along with $k=0$ to the functionality $\Funct{Pow}$. If the protocol aborts at iteration $k$, then the simulator sends the extracted shares of $b$ along with $k$ to $\Funct{Pow}$. %Transcripts can be simulated using the hybrid simulators.
\end{proof}

\begin{theorem}\label{thm:divbn}
\protx{Div}, \protx{BN} securely realize $\Funct{Div}$, $\Funct{BN}$ respectively, with abort, in the presence of one malicious party in the $(\Funct{Mult},\Funct{Pow},\Funct{Prep})$-hybrid model.
\end{theorem}
\begin{proof}
\protx{Div}, \protx{BN} are sequential combinations of local computations and invocations of $\Funct{Mult}$. Simulation follows directly from composing the simulators and input extraction follows from the simulator of \protx{Pow}.
\end{proof}

\vspace{-0.8cm}
\subsubsection{Protocol Overheads}\vspace{-3mm}
We theoretically estimate the overheads of our protocols in Table~\ref{tab:overheads} in Appendix~\ref{app:overhead}. 
The dominant round complexity for private compare comes from the string multiplication in Step~\ref{code:PC:blind}. $\wa{3}$ requires one additional round and one additional ring element (two in malicious security) over private compare.
Computing derivative of ReLU is a local computation over the $\wa{3}$ function. Computing ReLU requires two additional rounds and one ring element (two for malicious). Maxpool and derivative of require rounds proportional to the area of the filter. Finally, pow, division, and batch-norm requires a quadratic number of rounds in $\ell$.

\vspace{-0.8cm}

\section{Experimental Evaluation}\label{sec:eval}\vspace{-2mm}

We evaluate the performance of training and inference with \falcon on 6  networks of varying parameter sizes trained using MNIST, CIFAR-10 and Tiny ImageNet datasets (cf. Appendix~\ref{app:datasets}). A number of prior works such as SecureML~\cite{secureml}, DeepSecure~\cite{deepsecure}, MiniONN~\cite{minionn}, Gazelle~\cite{gazelle}, SecureNN~\cite{securenn}, ABY$^3$~\cite{aby3}, and Chameleon~\cite{chameleon} evaluate over these networks and we mimic their evaluation set-up for comparison. 

\vspace{-6mm}
\subsection{Experimental Setup}\label{subsec:setup}\vspace{-3mm}
We implement \falcon framework in about 14.6k LOC in {\tt C++} using the communication backend of SecureNN and will be open-sourced at \url{https://github.com/snwagh/falcon-public}. We run our experiments on Amazon EC2 machines over Ubuntu 18.04 LTS with Intel-Core i7 processor and 64GB of RAM. Our evaluation set-up uses similar as compared to prior work~\cite{securenn, aby3, secureml, chameleon, gazelle}.
We perform extensive evaluation of our framework in both the LAN and WAN setting. For the LAN setting, our bandwidth is about 625 MBps and ping time is about 0.2ms. For WAN experiments, we run servers in different geographic regions with 70ms ping time and 40 MBps bandwidth.

\begin{table*}[t]
\centering
%\resizebox{0.7\textwidth}{!}{
\begin{tabular}{c l c c c c c c c c c c c c c c}
%\toprule
& \multirow{2}{*}{Framework} & \multirow{2}{*}{Threat Model} &\multirow{2}{*}{LAN/ WAN}  & \multicolumn{2}{c}{Network-A} & \multicolumn{2}{c}{Network-B} & \multicolumn{2}{c}{Network-C} \\  \cmidrule{5-10}
& & & & Time  & Comm. & Time  & Comm. & Time  & Comm. \\ \midrule
\multirow{6}{*}{2PC} & SecureML~\cite{secureml}		& Semi-honest & LAN & $4.88$ & - & - & - & - & - \\
& DeepSecure~\cite{deepsecure}	& Semi-honest & LAN & - & - & $9.67$ & $791$ & - & -  \\
& EzPC~\cite{ezpc}	& Semi-honest & LAN & $0.7$ & $76$ & $0.6$ & $70$ & $5.1$ & $501$ \\
& Gazelle~\cite{gazelle}				& Semi-honest & LAN & $0.09$ & $0.5$ & $0.29$ & $0.8$ & $1.16$ & $70$ \\
& MiniONN~\cite{minionn} 			& Semi-honest & LAN &  $1.04$ & $15.8$ & $1.28$ & $47.6$ & $9.32$ & $657.5$ \\
& XONN~\cite{xonn} 					& Semi-honest & LAN & $0.13$ & $4.29$ & $0.16$ & $38.3$ & $0.15$ & $32.1$ \\ \midrule
%& Delphi~\cite{delphi} 					& Semi-honest & hmmmmmmmm & $0.13$ & $4.29$ & $0.16$ & $38.3$ & $0.15$ & $32.1$ \\ \midrule
 & Chameleon~\cite{chameleon}	& Semi-honest & LAN & - & - & $2.7$ & $12.9$ & -  & -  \\
& ABY$^3$~\cite{aby3} 				& Semi-honest & LAN & $0.008$ & $0.5$ & $0.01$ & $5.2$ & - & -  \\
& SecureNN~\cite{securenn}		& Semi-honest & LAN & $0.043$ & $2.1$ & $0.076$ & $4.05$ & $0.13$ & $8.86$  \\
\rowcolor{gray!20} \cellcolor{white} & & Semi-honest & LAN & $0.011$ & $0.012$ & $0.009$ & $0.049$ & $0.042$ & $0.51$ \\
\rowcolor{gray!20} \cellcolor{white}  & \multirow{-2}{*}{\ThisWork{}} & Malicious & LAN & $0.021$ & $0.31$ & $0.022$ & $0.52$ & $0.089$ & $3.37$ \\
& SecureNN~\cite{securenn}		& Semi-honest & WAN & $2.43$ & $2.1$ & $3.06$ & $4.05$ & $3.93$ & $8.86$  \\
\rowcolor{gray!20} \cellcolor{white} & & Semi-honest & WAN & $0.99$ & $0.012$ & $0.76$ & $0.049$ & $3.0$ & $0.5$\\
\rowcolor{gray!20} \cellcolor{white} \multirow{-8}{*}{3PC} & \multirow{-2}{*}{\ThisWork{}} & Malicious	& WAN & $2.33$ & $0.31$ & $1.7$ & $0.52$ & $7.8$ & $3.37$ \\ \midrule %\bottomrule
\multirow{2}{*}{4PC} & FLASH~\cite{flash}					& Malicious & LAN & $0.029$ & - & - & - & -  & -  \\
 & FLASH~\cite{flash}					& Malicious & WAN & $12.6$ & - & - & - & -  & -  \\
\end{tabular}%}
%\vspace{1mm}
\caption{\footnotesize Comparison of inference time of various frameworks for different networks using MNIST dataset. All runtimes are reported in seconds and communication in MB. ABY$^3$ and XONN do no implement their maliciously secure versions. 2-party (2PC) protocols are presented here solely for the sake of comprehensive evaluation of the literature.}
\label{tab:inference_sec}
\vspace{-0.5cm}
\end{table*}

\begin{table*}[t]
\centering
\resizebox{\textwidth}{!}{
\begin{tabular}{l c c c c c c c c c c c c c c}
%\toprule
\multirow{2}{*}{Framework} & \multirow{2}{*}{Threat Model} &\multirow{2}{*}{LAN/WAN}  & \multicolumn{2}{c}{LeNet (MNIST)} & \multicolumn{2}{c}{AlexNet (CIFAR-10)} & \multicolumn{2}{c}{VGG16 (CIFAR-10)} & \multicolumn{2}{c}{AlexNet (ImageNet)} & \multicolumn{2}{c}{VGG16 (ImageNet)}  \\  \cmidrule{4-13}
& & & Time  & Comm. & Time  & Comm. & Time  & Comm. & Time  & Comm. & Time  & Comm. \\ \midrule
\rowcolor{gray!20} & Semi-honest 		& LAN & $0.047$ & $0.74$ & $0.043$ & $1.35$ & $0.79$ & $13.51$ & $1.81$ & $19.21$ & $3.15$ & $52.56$ \\
 \rowcolor{gray!20} & Malicious 			& LAN & $0.12$ & $5.69$ & $0.14$ & $8.85$ & $2.89$ & $90.1$ & $6.7$ & $130.0$ & $12.04$* & $395.7$* \\
 \rowcolor{gray!20} & Semi-honest 		& WAN & $3.06$ & $0.74$ & $0.13$ & $1.35$ & $1.27$ & $13.51$ & $2.43$ & $19.21$ & $4.67$ & $52.56$\\
 \rowcolor{gray!20} \multirow{-4}{*}{\ThisWork{}} & Malicious 			& WAN & $7.87$ & $5.69$ & $0.41$ & $8.85$ & $4.7$ & $90.1$ & $8.68$ & $130.0$ & $37.6$* & $395.7$* \\ %\bottomrule
\end{tabular}}
\caption{\footnotesize Comparison of inference time of various frameworks over popular benchmarking network architectures from the machine learning domain. All runtimes are reported in seconds and communication in MB. * indicate non-amortized numbers.}
\label{tab:inference_ml}
\vspace{-0.5cm}
\end{table*}

\textbf{Optimizations:} All data-independent computation, i.e., pre-computation, is parallelized using 16 cores to reduce the run-time. When a ReLU layer is followed by a Maxpool layer, we swap the order of these two layers for optimized runtimes. We use the Eigen library for faster matrix multiplication and parallelize the private compare computation. We optimize across the forward and backward pass for Maxpool, ReLU, and Batch-Normalization layers, i.e., we compute the relevant derivatives while computing the functions. We use 32-bit integer range with 16 bits of fixed-point precision. As the entire codebase is parallelizable, significant improvement is possible by implementing \falcon using TensorFlow or PyTorch which support easy parallelization as well as computations over GPUs.

\textbf{Networks and Datasets:} For comparison with different networks as well as plaintext computations, we select 3 standard benchmarking datasets --- MNIST~\cite{mnist}, CIFAR-10~\cite{cifar}, and Tiny ImageNet~\cite{tinyimagenet} and 6 standard network architectures -- 3 from the privacy-preserving ML community and 3 from the ML community. For more details refer to Appendix~\ref{app:datasets},~\ref{app:networks}.

\vspace{-0.6cm}
\subsection{Results for Private Inference}\vspace{-2mm}
Tables~\ref{tab:inference_sec},~\ref{tab:inference_ml}  report the end-to-end latency time (in seconds) and number of bytes (in MB) communicated for performing a single inference query with
\ThisWork. We execute the queries in both LAN and WAN as well as semi-honest and malicious settings and compare with prior work wherever applicable.

\textbf{Comparison to Prior Work.}
We compare the inference time of a single query and the communication bytes of \falcon{} with prior work on networks A, B and C.
None of the prior works evaluate the remaining networks and hence we do not compare the performance of \falcon{} for the networks in Table~\ref{tab:inference_ml}.
Depending on the network architecture, our results are between $3\times$-$120\times$ faster than existing work. In particular, we are up to $18\times$ faster than XONN~\cite{xonn} ($11\times$ on average) and $32\times$ faster than Gazelle ($23\times$ on average), $8\times$ faster than SecureNN ($3\times$ on average), and comparable to ABY$^3$ on small networks.
%\sameer{Why are we selectively with only XONN and Gazelle? What about others ABY3?}
We are also $40\times$ more communication efficient than ABY$^3$~\cite{aby3}, $200\times$ more communication efficient than SecureNN~\cite{securenn}, and over $760\times$ more communication efficient compared to XONN~\cite{xonn}. \ch{Note that it is hard to compare frameworks without actually running benchmarks as different protocols scale differently for different architectures. For instance, GC based protocols scale better when run over WAN settings and larger networks change the fraction of the total overhead from linear layers for reasons described in Section~\ref{subsec:compute} and thus affect different protocols differently.} %\sameer{Communication improvement over ABY3?}

\textbf{Inference time and communication with \ThisWork.}
For both the adversarial settings, the inference latency for \falcon{}over LAN is within $25$ms for smaller networks (A and B) and around $100$ms for Network-C and LeNet. For AlexNet and VGG16, the inference time ranges from $0.5$ to $12$s depending on the model and the input dataset. The inference time increases with the size of the input image. Hence, queries over Tiny ImageNet are slower than CIFAR-10 for the same model architecture.
The inference time over the WAN setting ranges from $1$ to $3$s for the networks A, B and C and from $3$ to $37$s for the larger networks. However, we emphasize that the inference time with semi-honest adversarial setting is around $3\times$ faster than that for the malicious adversary. Hence, a faster deployment protocol is possible depending on the trust assumptions of the application.

In addition to efficient response times, our results show that \ThisWork{} requires small amounts of communication. Parties exchange less than $4$MB of data for smaller networks (Table~\ref{tab:inference_sec}) and $5$-$400$MB for larger networks (Table~\ref{tab:inference_ml}) (same for both LAN, WAN). However, similar to the inference time, malicious setting requires a higher communication and thus higher run-time.

\begin{table*}[h]
\centering
%\resizebox{0.7\textwidth}{!}{
\begin{tabular}{l c c c c c c c c c c c c c}
%\toprule
\multirow{2}{*}{Framework} & \multirow{2}{*}{Threat Model} &\multirow{2}{*}{LAN/ WAN}  &  \multicolumn{2}{c}{Network-A} & \multicolumn{2}{c}{Network-B} & \multicolumn{2}{c}{Network-C}  \\  \cmidrule{4-9}
& & & Time  & Comm. & Time  & Comm. & Time  & Comm. \\ \midrule
SecureML~\cite{secureml}* & Semi-honest 	& LAN & $81.7$ & - & - & - & - & - \\
SecureML~\cite{secureml}	& Semi-honest	& LAN & $7.02$ & - & - & - & - & - \\
ABY$^3$~\cite{aby3} & Semi-honest				& LAN & $0.75$ & $0.031$ & - & - & - & - \\ 
%XONN~\cite{xonn}				& & & & & & \\ 
SecureNN~\cite{securenn} 	 & Semi-honest & LAN		& $1.03$ & $0.11$ & - & - & $17.4$ & $30.6$ \\
\rowcolor{gray!20} & Semi-honest 			& LAN & $0.17$ & $0.016$ & $0.42$ & $0.056$ & $3.71$ & $0.54$ \\
\rowcolor{gray!20} \multirow{-2}{*}{\ThisWork{}} & Malicious		& LAN & $0.56$ & $0.088$ & $1.17$ & $0.32$ & $11.9$ & $3.29$ \\
SecureML~\cite{secureml}* & Semi-honest	& WAN	& $4336$ & - & - & - & - & - \\
SecureNN~\cite{securenn}	& Semi-honest  & WAN 	& $7.83$ & $0.11$ & - & - & $53.98$ & $30.6$ \\
\rowcolor{gray!20} & Semi-honest 			& WAN & $3.76$ & $0.016$ & $3.4$ & $56.14$ & $14.8$ & $0.54$ \\
\rowcolor{gray!20} \multirow{-2}{*}{\ThisWork{}} & Malicious 			& WAN & $8.01$ & $0.088$ & $7.5$ & $0.32$ & $39.32$ & $3.29$ \\ \midrule

Batch Size, Epochs &  &
& \multicolumn{2}{c}{128, 15} 
& \multicolumn{2}{c}{128, 15} 
& \multicolumn{2}{c}{128, 15} \\ %\bottomrule

%Batch Size, Epochs & \multicolumn{2}{c}{$128, 15$} &  \multicolumn{2}{c}{$128, 15$} & \multicolumn{2}{c}{$128, 15$} & \multicolumn{2}{c}{$128, 15$} & \multicolumn{2}{c}{$128, 15$} & \multicolumn{2}{c}{$128, 15$} \\
%Accuracy & \multicolumn{2}{c}{$99\%$} & \multicolumn{2}{c}{$99\%$} & \multicolumn{2}{c}{$99\%$} & \multicolumn{2}{c}{$99\%$} & \multicolumn{2}{c}{$99\%$} & \multicolumn{2}{c}{$99\%$} \\ \bottomrule

%\multirow{2}{*}{\begin{tabular}[c]{@{}c@{}}Batch Size, Epochs \\ Accuracy \end{tabular}} 
%& \multicolumn{2}{c}{\multirow{2}{*}{128, 15, 99\%}} 
%& \multicolumn{2}{c}{\multirow{2}{*}{128, 15, 99\%}} 
%& \multicolumn{2}{c}{\multirow{2}{*}{128, 15, 99\%}} 
%& \multicolumn{2}{c}{\multirow{2}{*}{128, 15, 99\%}} 
%& \multicolumn{2}{c}{\multirow{2}{*}{128, 15, 99\%}} 
%& \multicolumn{2}{c}{\multirow{2}{*}{128, 15, 99\%}} \\ 
%& & & & & & &\\ \bottomrule
\end{tabular}%}
%\vspace{1mm}
\caption{\footnotesize Comparison of training time of various frameworks over popular benchmarking network architectures from the security domain. All runtimes are reported in hours and communication in TB. * correspond to 2PC numbers. ABY$^3$ does not implement their maliciously secure protocols.}
\label{tab:training_sec}
\vspace{-0.5cm}
\end{table*}

\begin{table*}[h]
\centering
\resizebox{\textwidth}{!}{
\begin{tabular}{l c c c c c c c c c c c c c c c c}
%\toprule
\multirow{2}{*}{Framework} &\multirow{2}{*}{Threat Model} &\multirow{2}{*}{LAN/ WAN}  & \multicolumn{2}{c}{LeNet} & \multicolumn{2}{c}{AlexNet (CIFAR-10)} & \multicolumn{2}{c}{VGG16 (CIFAR-10)} &
\multicolumn{2}{c}{AlexNet (ImageNet)} & \multicolumn{2}{c}{VGG16 (ImageNet)} \\ 
\cmidrule{4-13}
& & & Time  & Comm. & Time  & Comm. & Time  & Comm. & Time & Comm. & Time & Comm. \\ \midrule
\rowcolor{gray!20} & Semi-honest 	 			& LAN & $6.05\times 10^0$ & $0.81$ & $7.89\times 10^1$ & $7.24$ & $8.43\times 10^2$ & $45.9$ & $1.23\times 10^4$ & $222.9$ & $5.19\times 10^3$ & $156.0$ \\
\rowcolor{gray!20} & Malicious 		 			& LAN & $1.22\times 10^1$ & $4.82$ & $2.82\times 10^2$ & $43.4$ & $3.05\times 10^3$ & $185.3$ & $4.63\times 10^4$ & $1598$ & $1.95\times 10^4$ & $1012$ \\
\rowcolor{gray!20} & Semi-honest 	 			& WAN & $1.85\times 10^1$ & $0.81$ & $2.33\times 10^2$ & $7.24$ & $2.09\times 10^3$ & $45.9$ & $1.54\times 10^4$ & $222.9$ & $6.89\times 10^3$ & $156.0$ \\
\rowcolor{gray!20} \multirow{-4}{*}{\ThisWork{}} &Malicious 	& WAN & $5.20\times 10^1$ & $4.82$ & $7.24\times 10^2$ & $43.4$ & $5.26\times 10^3$ & $185.3$ & $5.71\times 10^4$ & $1598$ & $2.47\times 10^4$ & $1012$ \\ \midrule
\multicolumn{3}{c}{Batch Size, Epochs}
& \multicolumn{2}{c}{128, 15} 
& \multicolumn{2}{c}{128, 90} 
& \multicolumn{2}{c}{128, 25}
& \multicolumn{2}{c}{128, 90} 
& \multicolumn{2}{c}{128, 25} \\ %\bottomrule
\end{tabular}}
\caption{\footnotesize Comparison of training time of various frameworks over popular benchmarking network architectures from the machine learning domain. All runtimes are \ch{reported in hours and communication in TB.}}
\label{tab:training_ml}
\vspace{-0.8cm}
\end{table*}

%\begin{table*}[h]
%\centering
%\resizebox{\textwidth}{!}{
%\begin{tabular}{l c c c c c c c c c c c c c c c c}
%%\toprule
%\multirow{2}{*}{Framework} &\multirow{2}{*}{Threat Model} &\multirow{2}{*}{LAN/ WAN}  & \multicolumn{2}{c}{LeNet} & \multicolumn{2}{c}{AlexNet (CIFAR-10)} & \multicolumn{2}{c}{VGG16 (CIFAR-10)} &
%\multicolumn{2}{c}{AlexNet (ImageNet)} & \multicolumn{2}{c}{VGG16 (ImageNet)} \\ 
%\cmidrule{4-13}
%& & & Time  & Comm. & Time  & Comm. & Time  & Comm. & Time & Comm. & Time & Comm. \\ \midrule
%\rowcolor{gray!20} & Semi-honest 	 			& LAN & $0.036$ & $0.81$ & $0.47$ & $7.24$ & $5.02$ & $45.9$ & $73.0$ & $222.9$ & $30.9$ & $156.0$ \\
%\rowcolor{gray!20} & Malicious 		 			& LAN & $0.73$ & $4.82$ & $1.68$ & $43.4$ & $18.18$ & $185.3$ & $276$ & $1598$ & $116$ & $1012$ \\
%\rowcolor{gray!20} & Semi-honest 	 			& WAN & $0.11$ & $0.81$ & $1.39$ & $7.24$ & $12.44$ & $45.9$ & $91.7$ & $222.9$ & $41.0$ & $156.0$ \\
% \rowcolor{gray!20} \multirow{-4}{*}{\ThisWork{}} &Malicious 	& WAN & $0.31$ & $4.82$ & $4.31$ & $43.4$ & $31.32$ & $185.3$ & $340$ & $1598$ & $147$ & $1012$ \\ \midrule
%\multicolumn{3}{c}{Batch Size, Epochs}
%& \multicolumn{2}{c}{128, 15} 
%& \multicolumn{2}{c}{128, 90} 
%& \multicolumn{2}{c}{128, 25}
%& \multicolumn{2}{c}{128, 90} 
%& \multicolumn{2}{c}{128, 25} \\ %\bottomrule
%\end{tabular}}
%\caption{\footnotesize Comparison of training time of various frameworks over popular benchmarking network architectures from the machine learning domain. All runtimes are reported in weeks and communication in TB.}
%\label{tab:training_ml}
%\vspace{-0.8cm}
%\end{table*}

\vspace{-0.5cm}
\subsection{Results for Private Training}\vspace{-2mm}
Tables~\ref{tab:training_sec},~\ref{tab:training_ml} report the execution time and the communication required for training  the $6$ networks.

\textbf{Comparison to Prior Work.}
For private training, \falcon is up to $6\times$ faster than SecureNN~\cite{securenn} ($4\times$ on average), $4.4\times$ faster than ABY$^3$ and $70\times$ faster than SecureML~\cite{secureml}. We highlight that \falcon{} achieves these speedups due to improved protocols (both round complexity and communication as described in Section~\ref{subsec:techcontrib}).  Table~\ref{tab:training_sec} shows that the communication overhead is $10\times$ to $100\times$ as compared to other solutions.

\textbf{Execution time for \ThisWork.} The time to privately train networks A, B and C with \falcon is around $3$ to $40$ hrs.
For larger networks, we extrapolate time from a single iteration of a forward and a backward pass. The training time ranges from a few weeks to hundreds of weeks. Although these values seem to be quite large, high capacity machine learning models are known to take from a few days to weeks to achieve high accuracy when trained (both on CPU as well as GPU). Such networks can also benefit from transfer learning techniques, where a public pre-trained model is fine-tuned with a private dataset. This fine-tuning requires fewer epochs and hence speed up the overall runtime considerably.

\vspace{-5mm}
\subsection{Compute vs. Communication Cost}\label{subsec:compute}\vspace{-2mm}
Figure~\ref{fig:compute_commute} shows the computation time as compared to the communication time for the inference of a single input over different network sizes. We observe that the computation cost increases with the network size and becomes the dominant reason for the performance overhead in private deep learning with \ThisWork{}. %As seen in Figure~\ref{fig:compute_commute}, for larger networks, the computation costs dominate. 
\ch{The reason for this is because the complexity of matrix multiplication is ``super-quadratic'' i.e., to multiply two $n\times n$ matrices, the computation overhead is strictly larger than $O(n^2)$. Note that the communication of the matrix multiplication protocol in this work is only linear in the size of the matrices and has a round complexity of a single round. On the other hand, the non-linear operations, though more communication expensive in MPC, are applied on vectors of size equal to the output of the matrix product and thus are ``quadratic.'' In other words, the non-linear operations such as ReLU are applied on the output of the matrix multiplication (FC/Conv layers) and are applied on vectors of size $O(n^2)$ assuming they are applied on the output of the multiplication of two $n\times n$ matrices. Hence, for large network architectures, the time required for the matrix-multiplication dominates the overall cost.}

This observation is against the  conventional wisdom that MPC protocols are communication bound and not computation bound. When running larger networks such as AlexNet and VGG16, and especially for Tiny ImageNet, the computation time starts becoming a significant fraction of the total time. Hence, we claim that \falcon{} is optimized for communication rounds, specifically when operating over large networks. With our results, we motivate the community to focus on designing faster compute solutions using accelerators such as GPUs, parallelization, efficient matrix multiplications and caching, along with the conventional goals of reducing communication and round complexity.

\begin{figure}[h]
\centering
\includegraphics[width = 1\columnwidth]{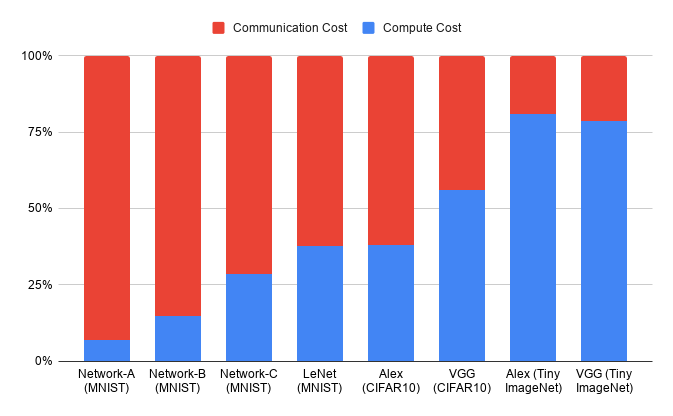}
\caption{\footnotesize Compute vs. communication cost for private inference using \falcon over WAN  for the malicious adversary. We show that as the network size increases, computation becomes a dominant factor in the overall end-to-end runtime.}
\label{fig:compute_commute}
\vspace{-0.4cm}
\end{figure}

% set up labelformat and labelsep for subfigure
\captionsetup[subfigure]{labelformat=parens, singlelinecheck=on}

\begin{figure*}[htp]
\centering
\begin{subfigure}[t]{0.3\textwidth}
		\centering
		\includegraphics[width=\textwidth]{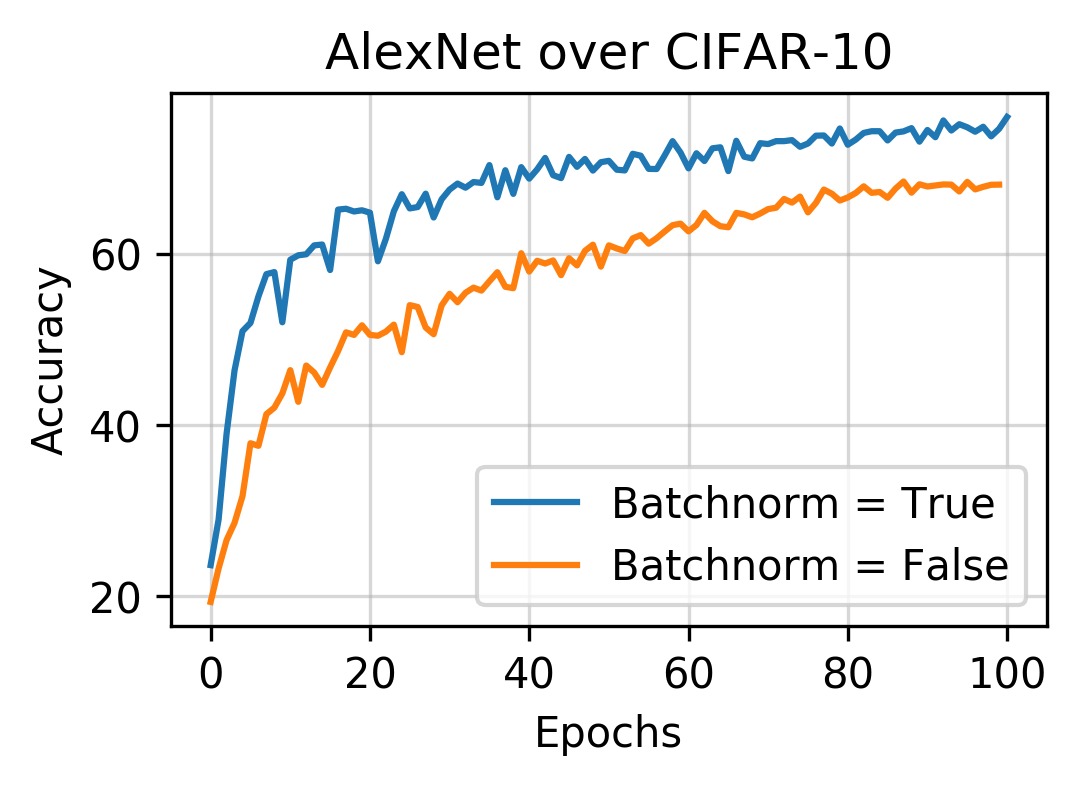}
		\caption{}\label{fig:BNa}
	\end{subfigure}
	\quad
\begin{subfigure}[t]{0.3\textwidth}
		\centering
		\includegraphics[width=\textwidth]{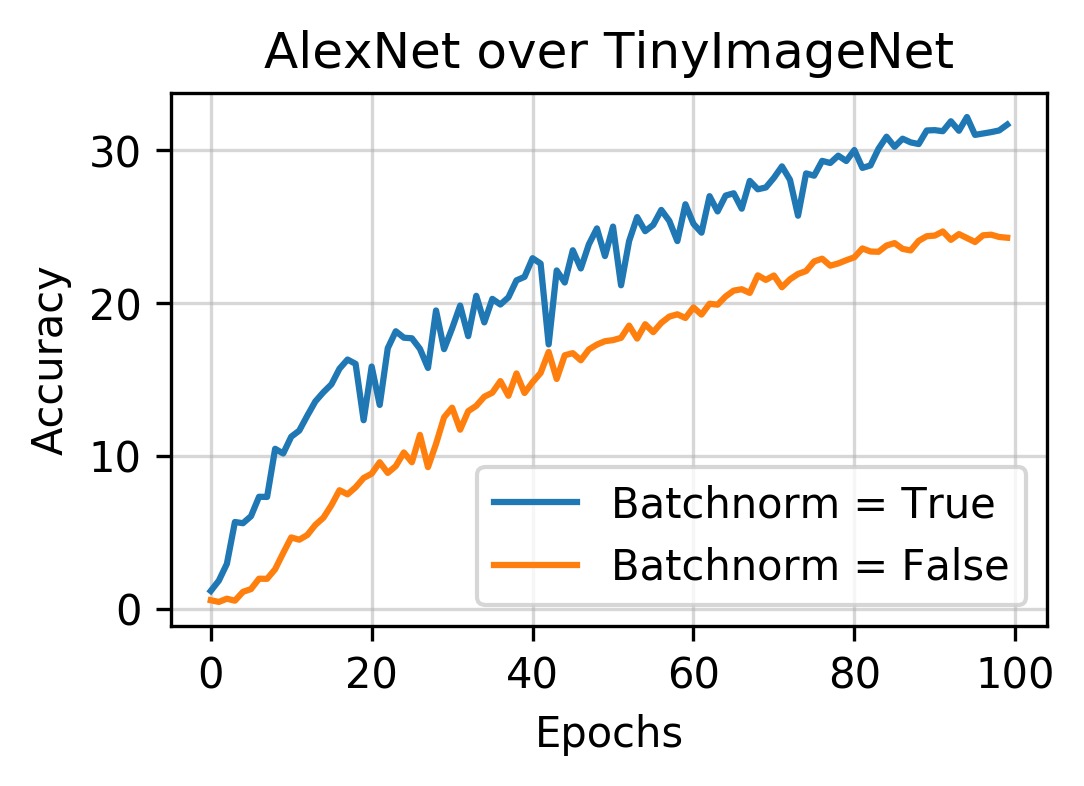}
		\caption{}\label{fig:BNb}
	\end{subfigure}
	\quad
\begin{subfigure}[t]{0.3\textwidth}
		\centering
		\includegraphics[width=\textwidth]{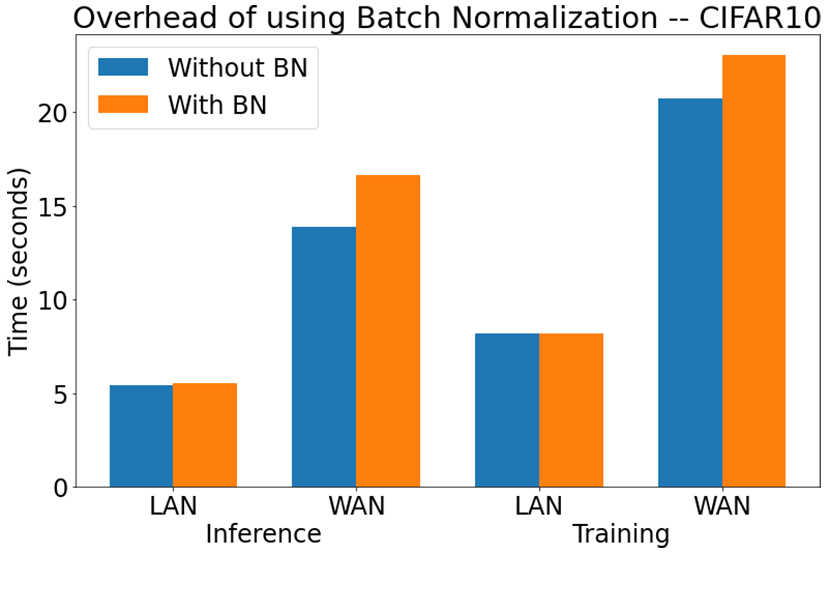}
		\caption{}\label{fig:BNc}
	\end{subfigure}
\caption{\footnotesize \ch{In Figs.~\ref{fig:BNa},~\ref{fig:BNb}, we study the model accuracy with and without batch normalization layers as a function of epochs for AlexNet network. As can be seen, batch normalization not only helps train the network faster but also train better networks. In Fig.~\ref{fig:BNc}, we study the performance overhead of running the network (using \ThisWork{}) with and without batch normalization layers.}}
\label{fig:BN}
\vspace{-4mm}
\end{figure*}

%%%%%%%%%%%%%%%%%%%%%%%%%%%%%%%%%%%%%%%%
\begin{table}[t]
\centering
\resizebox{\columnwidth}{!}{
\begin{tabular}{c @{\hskip 0.2in}  c c c c}
\mr{2}{Network} 	& Training		& Inference		& \ThisWork{} Inference		& Relative \\
							& Accuracy	& Accuracy		& Accuracy	& Error \\ \cmidrule(lr){1-5}

Network-A 			&	98.18\%	&	97.42\%			&	97.42\%					& 0.471\%\\			%SecureML
Network-B 			&	98.93\%	&	97.81\%			&	97.81\%					& 0.635\% \\			% Sarda
Network-C 			&	99.16\%		&	98.64\%		&	98.64\%				& 0.415\% \\			% MiniONN
LeNet 					&	99.76\%	&	99.15\%			&	96.85\%				& 0.965\% \\ 			% LeNet
\end{tabular}}
\caption{\ch{Summary of experiments involving accuracy of neural networks using secure computation. The first two columns refer to the plaintext accuracies and relative error refers to the average relative error of one forward pass computation using \ThisWork.} }
\label{tab:accuracy}
\vspace{-0.8cm}
\end{table}

\vspace{-0.5cm}
\subsection{Comparison vs. Plaintext Computation}\vspace{-3mm}

Given the surprising insights from Figure~\ref{fig:compute_commute}, we also compare the execution of privacy-preserving computations with plaintext computations. These results are summarized in Table~\ref{tab:plaintext}. We use standard PyTorch libraries for the plaintext code, similar hardware as that of privacy-preserving benchmarks for CPU-CPU comparison, and use a single Nvidia P100 GPU for the GPU-CPU comparison. 
%We compare our private deep learning overheads with plaintext execution of the same networks. 
Our findings indicate that private deep learning (over CPU) is within a factor of $40\times$-$1200\times$ of plaintext execution of the same network over CPU and within $50\times$-four orders  of magnitude that of plaintext execution over GPU (using PyTorch) when performed over LAN. The overhead further increases by $1.2\times$-$2.4\times$ when comparing against WAN evaluations. This indicates the importance of supporting GPUs and optimizers for private deep learning and showcases the need for further reducing the overhead of MPC protocols. We believe that it is beneficial for the broader research community to have an estimate of the gap between plaintext and privacy-preserving techniques for realistic size networks and datasets.

\begin{table*}[h]
\centering
\resizebox{\textwidth}{!}{
\begin{tabular}{c c c c c c c c c c c}
& \mc{2}{\mr{3}{Run-type}} & \mc{4}{CIFAR-10} & \mc{4}{Tiny ImageNet} \\ \cmidrule(lr){4-7} \cmidrule(lr){8-11}
& & & \mc{2}{Training} & \mc{2}{Inference} & \mc{2}{Training} & \mc{2}{Inference} \\ \cmidrule(lr){4-5} \cmidrule(lr){6-7} \cmidrule(lr){8-9} \cmidrule(lr){10-11}
& & & AlexNet & VGG-16 & AlexNet & VGG-16 & AlexNet & VGG-16 & AlexNet & VGG-16 \\  \midrule
\mr{2}{Plaintext} & CPU-only & localhost 			& 	$1.6\times 10^2$	&	$7.3\times 10^2$	& $7.2 \times 10^1$   &	$3.4\times 10^2$  & 
														$5.0\times 10^2$	&	$3.1\times 10^3$	& $2.5 \times 10^2$   &	$1.3\times 10^3$  \\

& GPU-assisted	& localhost								& $2.8\times 10^1$		&	$6.4\times 10^1$	& $3.8\times 10^1$   &	$5.8\times 10^1$   &
														$3.6\times 10^1$	&	$1.2\times 10^2$ 	& $3.8\times 10^1$		&	$5.7\times 10^1$ \\ \cmidrule(lr){1-11}
														
\mr{2}{Private} & CPU-only  & LAN			& 	$6.4\times 10^3$	&	$2.5\times 10^5$	& $5.6 \times 10^3$   &	$1.0\times 10^5$  & 
														$6.3\times 10^5$	&	$9.5\times 10^5$	& $2.3 \times 10^5$   &	$4.0\times 10^5$ \\% &  $\leftarrow$ LAN \\
& CPU-only & WAN 										& 	$2.4\times 10^4$	&	$6.2\times 10^5$	& $1.7 \times 10^4$   &	$1.6\times 10^5$  & 
														$7.8\times 10^5$	&	$1.2\times 10^6$	& $3.1 \times 10^5$   &	$5.9\times 10^5$ \\%&  $\leftarrow$ WAN \\
													\midrule
Private & \mc{2}{Bandwidth}			& 	$6.4\times 10^3$	&	$2.5\times 10^5$	& $5.6 \times 10^3$   &	$1.0\times 10^5$  & 
														$6.3\times 10^5$	&	$9.5\times 10^5$	& $2.3 \times 10^5$   &	$4.0\times 10^5$ \\	
\end{tabular}}
\caption{Comparison of private computation (for semi-honest protocols, cf Section~\ref{subsec:setup} for network parameters) with plaintext over the same hardware using PyTorch and a single NVIDIA P100 GPU. Numbers are for a 128 size batch in milliseconds.}
\label{tab:plaintext}
\vspace{-0.7cm}
\end{table*}

\vspace{-0.5cm}
\ch{\subsection{Batch Normalization and Accuracy}\vspace{-3mm}

We study the benefits of batch normalization for privacy-preserving training of neural networks. We compute the accuracy of partially trained models after each epoch with and without the batch normalization layers. As seen in Figs.~\ref{fig:BNa},~\ref{fig:BNb}, batch normalization layers not only help train the network faster but also train better networks. Fig.~\ref{fig:BNc} demonstrates the overhead of MPC protocols with and without batch normalization layers. Given the high round complexity of batch normalization, the gap is significant only in the WAN setting.

We also study the effect of our approximations and smaller datatype on the accuracy of the computation. We compare the evaluation of the networks with 64-bit {\tt float} datatypes over PyTorch against a 32-bit datatype {\tt uint32\_t} using fixed-point arithmetic for \ThisWork. The final layer outputs differ by small amounts (less than 1\%) in comparison with the high precision 64-bit computation. \chh{As a consequence, as seen in Table~\ref{tab:accuracy}, most networks show no/low loss in the overall neural network accuracy when the computation is performed as fixed-point integers over 32-bit datatype. This is because the final prediction is robust to small relative error in individual values at the output. This also makes the final prediction vector inherently noisy and may provide some defense against model inversion attacks.}

}

\vspace{-0.8cm}
\section{Related Work}\label{sec:related}\vspace{-3mm}
\textbf{Privacy-preserving training.} In a seminal paper on private machine learning, Mohassel~\etal~\cite{secureml} show protocols for a variety of machine learning algorithms such as linear regression, logistic regression and neural networks. Their approach is based on a 2-party computation model and rely on techniques such as oblivious transfer~\cite{OT} and garbled circuits~\cite{GC}. Following that, Mohassel~\etal~\cite{aby3} proposed a new framework called ABY$^3$ which generalizes and optimizes switching back and forth between arithmetic, binary, and Yao garbled circuits in a 3-party computation model. \comment{ABY$^3$ extends the 2-party ABY framework~\cite{ABY} into a 3-party setting and improves the performance of state-of-the-art multi-party computation. ABY$^3$ extend the techniques from Araki~\etal~\cite{AFLNO16} to arithmetic secret sharing. }
Wagh~\etal~\cite{securenn} proposed SecureNN that considers a similar 3-party model with semi-honest security and eliminate expensive cryptographic operations to demonstrate privacy-preserving training and inference of neural networks. SecureNN also provides {\em malicious privacy}, a notion formalized by Araki~\etal~\cite{AFLNO16} but not correctness under malicious corruption.
\falcon provides a holistic framework for both training and inference of neural networks while improving computation and communication overhead as compared to prior work.

\textbf{Privacy-preserving inference.} Privacy-preserving inference has received considerable attention over the last few years. Recall that we have summarized some of these works in Table~\ref{tab:comparison}. Private inference typically relies on one or more of the following techniques: secret sharing~\cite{securenn, aby3}, garbled circuits~\cite{xonn, ezpc}, homomorphic encryption~\cite{minionn, secureml, helen} or Goldreich-Micali-Wigderson (GMW)~\cite{gmw87, chameleon}, each with its own advantages and disadvantages.
CryptoNets~\cite{cryptonets} was one of the earliest works to demonstrate the use of homomorphic encryption to perform private inference. CryptoDL~\cite{2017arXiv171105189H} developed techniques that use approximate, low-degree polynomials to implement non-linear functions and improve on CryptoNets. DeepSecure~\cite{deepsecure} uses garbled circuits to develop a privacy-preserving deep learning framework.

%SecureML~\cite{secureml} uses homomorphic encryption, garbled circuits, and secret sharing to perform private inference. ABY$^3$ develops a mixed protocol framework and switches between arithmetic, binary and garbled circuits to achieve the same goals. 
Chameleon~\cite{chameleon} is another mixed protocol framework that uses the Goldreich-Micali-Wigderson (GMW) protocol~\cite{gmw87} for low-depth non-linear functions, garbled circuits for high-depth functions and secret sharing for linear operations to achieve high performance gains. The above three~\cite{secureml, aby3, chameleon} demonstrate private machine learning for other machine learning algorithms such as SVMs, linear and logistic regression as well. Gazelle~\cite{gazelle} combines techniques from homomorphic encryption with MPC and optimally balances the use of a specially designed linear algebra kernel with garbled circuits to achieve fast private inference. EzPC~\cite{ezpc} is a ABY-based~\cite{ABY} secure computation framework that translates high-level programs into Boolean and arithmetic circuits. Riazi~\etal{} propose a framework XONN~\cite{xonn} and showcase compelling performance for inference on large {\em binarized} neural networks and uses garbled circuits to provide constant round private inference. The work also provides a simple easy-to-use API with a translator from Keras~\cite{keras} to XONN. EPIC~\cite{epic} demonstrates the use of transfer learning in the space of privacy-preserving machine learning while Quotient~\cite{quotient} takes the first steps in developing two party secure computation protocols for optimizers and normalizations. CrypTFlow~\cite{cryptflow} builds on SecureNN and uses trusted hardware to achieve maliciously secure protocols in a 3PC model. Delphi~\cite{delphi} builds on Gazelle to further improve performance and proposes a novel planner that automatically generates neural network architecture configurations that navigate the performance-accuracy trade-offs. 
Astra~\cite{astra} is a 3PC protocol with semi-honest security and forms the foundation for a few follow-up works. BLAZE~\cite{blaze} builds on Astra to achieve malicious security and fairness in a 3PC honest majority corruption model and uses an adder circuit approach for non-linear function computation. Trident achieves the same result in a 4PC model with further performance improvements. FLASH~\cite{flash} also proposes a 4PC model that achieves malicious security with guaranteed output delivery. 
QuantizedNN~\cite{quantizednn} proposes an efficient PPML framework using the quantization scheme of Jacob~\etal~\cite{jacobquantization} and provides protocols in all combinations of semi-honest/malicious security and honest majority vs dishonest majority corruptions. 
%Recent works~\cite{practicalfully, boneh2019zero} explore a stronger adversarial model called guaranteed output delivery in a similar 3 party honest majority corruption model.
%Alternatively, protocols such as BDOZa~\cite{bdoza} and SPDZ~\cite{spdz} operate in a ``dishonest-majority'' adversarial model and allow $n-1$ out of the $n$ computing parties to be corrupted. Helen~\cite{helen} allows privacy-preserving distributed training of linear models in a similar setting. 
%\comment{Follow-up works include Topgear~\cite{topgear}, Overdrive~\cite{overdrive} which improve the efficiency of zero-knowledge proofs, MASCOT~\cite{mascot} which uses efficient pre-processing using oblivious transfers, and Marbled Circuits~\cite{marbled} which uses efficient share conversions for Boolean shares.}
%\comment{
%\textbf{Other privacy-preserving analytics:} Other approaches to privacy-preserving use techniques such as differential privacy or frameworks similar to federated learning~\cite{federratedlearning}. Abadi~\etal~\cite{abadi2016deep} develop algorithmic techniques for learning within the framework of differential privacy. PATE~\cite{papernot2016semi} demonstrate private learning over several different ML models. Bonawitz~\etal~\cite{bonawitz2017practical} demonstrate a privacy-preserving approach in a federated learning setting to aggregate user-provided model updates for a deep neural network.}

\vspace{-0.6cm}
\vspace{-2mm}
\section{Conclusion}\label{sec:conclusion}
\vspace{-3mm}
\ThisWork{} supports new protocols for private training and inference in a honest-majority 3-party setting. Theoretically, we propose novel protocols that improve the round and communication complexity and provide security against maliciously corrupt adversaries with an honest majority. 
\ThisWork{} thus provides malicious security and provides several orders of magnitude performance improvements over prior work. Experimentally, \falcon is the first secure deep learning framework to examine performance over large-scale networks such as AlexNet and VGG16 and over large-scale datasets such as Tiny ImageNet. We also are the first work to demonstrate efficient protocols for batch-normalization which is a critical component of present day machine learning.

\vspace{-0.8cm}
\section*{Acknowledgments}
\vspace{-3mm}
We thank Vikash Sehwag for his help with the experiments, the anonymous reviews, and our Shepherd Melek \"{O}nen. We also thank the following grant/awards for supporting this work: Facebook Systems for ML award, Qualcomm Innovation Fellowship, Princeton CSML DataX award, Princeton E-ffiliates Partnership award, Army Research Office YIP award, Office of Naval Research YIP Award, and National Science Foundation's CNS-1553437 and CNS-1704105, ISF grant 2774/20, BSF grant 2018393, NSF-BSF grant 2015782, and a grant from the Ministry of Science and Technology, 
Israel, and the Dept. of Science and Technology, Government of India.

\vspace{-0.8cm}

\bibliographystyle{IEEEtran}
\bibliography{bib}

\appendix

%\section{Appendix}\label{sec:appendix}
%
%Here we provide some relevant technical details on the various layers of neural networks as well as the different network architectures used in this work.

\vspace{-0.6cm}
\section{Recent Related Work}
\vspace{-0.2cm}

We compare the theoretical complexities of protocols in Astra, BLAZE, FLASH, and Trident with protocols in \ThisWork{}. Since the approach for computing non-linear operations is fundamentally different, we compare the end-to-end overhead of the ReLU protocol in each of these frameworks. Table~\ref{tab:rebuttal} shows a comparison of the theoretical complexities. Note that BLAZE and Astra are 3PC protocols and FLASH and Trident are 4PC protocols. In terms of evaluation of neural networks, most of these works evaluate their approach only over DNNs. None of these frameworks evaluate their approaches for \textit{training} of neural networks. Comparison of concrete efficiency of these protocols is documented in Table~\ref{tab:inference_sec}.

\vspace{-0.6cm}
\section{Theoretical Complexity}
\label{app:overhead}
\vspace{-0.2cm}

We theoretically estimate the overheads of our protocols in Table~\ref{tab:overheads}. The dominant round complexity for private compare comes from the string multiplication in Step~\ref{code:PC:blind}. $\wa{3}$ requires one additional round and one additional ring element (two in malicious security) over private compare.
Computing derivative of ReLU is a local computation over the $\wa{3}$ function. Computing ReLU requires two additional rounds and one ring element (two for malicious). Maxpool and derivative of require rounds proportional to the area of the filter. Finally, pow, division, and batch-norm require a quadratic number of rounds in $\ell$.

\begin{table*}[h]
\centering
\resizebox{\textwidth}{!}{
\begin{tabular}{c@{\hskip 5mm} c c c c c c c c c c}
& \multicolumn{6}{c}{3PC} & \multicolumn{4}{c}{4PC} \\ \cmidrule(lr){2-7} \cmidrule(lr){8-11}
\multirow{2}{*}{Protocol} & \multicolumn{2}{c}{Astra} & \multicolumn{2}{c}{BLAZE} & \multicolumn{2}{c}{\sc Falcon} & \multicolumn{2}{c}{FLASH} & \multicolumn{2}{c}{Trident} \\  \cmidrule{2-11}
& Round  & Comm. & Round  & Comm. & Round  & Comm. & Round  & Comm. & Round  & Comm. \\ \midrule 
Multiplication & $1$ & $4 \ell$ & $1$ & $3\ell$ & $1$ & $4 \ell$ & $1$ & $3\ell$ & $1$ & $3\ell$ \\
ReLU           & $3 + \log \ell$ & $45 \ell$ & $4$ & $(\kappa + 7)\ell$ & $5 + \log \ell$ & $32 \ell$ & $\log \ell + 10$ & $46 \ell$ & $4$ & $8\ell + 2$ \\
\end{tabular}}
\caption{Comparison of theoretical complexities of privacy-preserving protocols for Multiplication and ReLU with Astra, BLAZE, FLASH, Trident, and {\sc Falcon}. $\ell$ is the bit-size of the datatype and $\kappa$ is the security parameter (usually set to 40).}
\label{tab:rebuttal}
\end{table*}

\begin{table*}[t]
\centering
\resizebox{\textwidth}{!}{
\begin{tabular}{c@{\hskip 6mm} c@{\hskip 2mm} c c c c c}
%\toprule
& \multirow{2}{*}{Protocol} & \multirow{2}{*}{Dependence} & \multicolumn{2}{c}{Semi-Honest} & \multicolumn{2}{c}{Malicious} \\ \cmidrule{4-7}
 & & & Rounds & Comm & Rounds & Comm \\ \midrule
 \parbox[t]{2mm}{\multirow{3}{*}{\rotatebox[origin=c]{90}{$\substack{\textrm{\textbf{Basic}}\\ \textrm{\textbf{Protocols}}}$}}} & MatMul & $(x\times y)(y\times z)$ & $1$ & $kxz$ & $1$ & ${2kxz}$ \\
    & Private Compare & $n$ & $2 + \log_2 \ell$ & $2kn$ & $2 + \log_2 \ell $ & $4kn$ \\
    & $\wa{3}$ & $n$ & ${3 + \log_2 \ell}$ & $3kn$ & ${3 + \log_2 \ell}$ & $6kn$  \\ \midrule
\parbox[t]{2mm}{\multirow{7}{*}{\rotatebox[origin=c]{90}{$\substack{\textrm{\textbf{Compound}}\\ \textrm{\textbf{Protocols}}}$}}} & ReLU and & \multirow{2}{*}{${n}$} & \multirow{2}{*}{${5 + \log_2 \ell}$} & \multirow{2}{*}{$4kn$} & \multirow{2}{*}{${5 + \log_2 \ell}$} & \multirow{2}{*}{$8kn$}\\
& Derivative of ReLU &  &  &  &  & \\
& MaxPool and & \multirow{2}{*}{${n, \{w, h\}}$} & \multirow{2}{*}{${(wh-1)}(7 + \log_2 \ell)$} & \multirow{2}{*}{$5k + wh$} & \multirow{2}{*}{${(wh-1)}{(7 + \log_2 \ell)}$} & \multirow{2}{*}{$10k + 2wh$}\\
& Derivative of Maxpool &  &  &  &  & \\ 
& Pow   & $n$ & $5 \ell + \ell \cdot \log_2 \ell$ & $4kn\ell$ & $5 \ell + \ell \cdot \log_2 \ell$ & $8kn\ell$ \\
& Division & $n$ & ${7 + 5 \ell + \ell \cdot \log_2 \ell}$ & $4kn\ell + 7kn$ & ${7 + 5 \ell + \ell \cdot \log_2 \ell}$ & $8kn\ell + 14kn$ \\ 
& Batch Norm & $r,n$ & ${15 + 5 \ell + \ell \cdot \log_2 \ell}$ & $kr + 4kr\ell + 14krn$ & ${15 + 5 \ell + \ell \cdot \log_2 \ell}$ & $2kr + 8kr\ell + 28krn$ \\ %\bottomrule
\end{tabular}}
\caption{\footnotesize Theoretical overheads of basic and compound protocols. Communication is in Bytes where $\ell$ is the logarithm of the ring size and $k$ is its Byte size. We use $n$ to denote the size of the vector in vectorized implementations. \chh{Malicious protocols suffer from higher communication complexity compared to semi-honest protocols that results in poor concrete efficiency when implemented.}}
\label{tab:overheads}
\vspace{-0.5cm}
\end{table*}

%\begin{table*}[ht]
%\centering
%\begin{tabular}{c@{\hskip 6mm} c@{\hskip 2mm} c c c c}
%& \multirow{2}{*}{Protocol} & \multicolumn{2}{c}{Semi-Honest} & \multicolumn{2}{c}{Malicious} \\ \cmidrule{3-6}
% & & Rounds & Comm & Rounds & Comm \\ \midrule
% \parbox[t]{2mm}{\multirow{3}{*}{\rotatebox[origin=c]{90}{$\substack{\textrm{\textbf{Basic}}\\ \textrm{\textbf{Protocols}}}$}}} & MatMul & $1$ & $k$ & $1$ & ${2*k}$ \\
%    & Private Compare & ${\ell +2}$ & $2*k$ & ${\ell +2}$ & $4*k$ \\
%    & $\wa{3}$ & ${\ell +3}$ & $3*k$ & ${\ell +3}$ & $6*k$  \\ \midrule
%\parbox[t]{2mm}{\multirow{5}{*}{\rotatebox[origin=c]{90}{$\substack{\textrm{\textbf{Compound}}\\ \textrm{\textbf{Protocols}}}$}}} & ReLU and & \multirow{2}{*}{${\ell +5}$} & \multirow{2}{*}{$4*k + 1/8$} & \multirow{2}{*}{${\ell +5}$} & \multirow{2}{*}{$8*k + 2/8$}\\
%& Derivative of ReLU &  &  &  & \\
%& MaxPool and & \multirow{2}{*}{${f_w*f_h - 1}*{\ell + 8}$} & \multirow{2}{*}{$5*k + 2/8 + f_w*f_h$} & \multirow{2}{*}{${(f_w*f_h - 1)}*{\ell + 8}$} & \multirow{2}{*}{$10*k + 4/8 + 2*f_w*f_h$}\\
%& Derivative of Maxpool &  &  &  &  \\ 
%& Division & ${3}$ & $3*k$ & ${3}$ & $6*k$ \\ \bottomrule
%\end{tabular}
%\vspace{3mm}
% \caption{\textbf{Theoretical overheads of basic and compound protocols. Communication is in Bytes where $\ell$ is the logarithm of the ring size and $k$ is its Byte size.}}
%\label{tab:overheads}
%\end{table*}

\vspace{-0.6cm}
\section{Datasets}\label{app:datasets}
\vspace{-0.2cm}

We select 3 datasets popularly used for training image classification models --- MNIST~\cite{mnist}, CIFAR-10~\cite{cifar}, and Tiny ImageNet~\cite{tinyimagenet}. We describe these below:

\begin{enumerate}[label=(\Alph*), itemsep=0pt, topsep=2pt, leftmargin=10mm]
    \item \textbf{MNIST~\cite{mnist}: } MNIST is a collection of handwritten digits dataset. It consists of 60,000 images in the training set and 10,000 in the test set. Each image is a $28 \times 28$ pixel image of a handwritten digit along with a label between 0 and 9. We evaluate Network-A, B, C, and the LeNet network on this dataset in both the semi-honest and maliciously secure variants.
    \item \textbf{CIFAR-10~\cite{cifar}: } CIFAR-10 consists of 60,000 images (50,000 training and 10,000 test images) of 10 different classes (such as airplanes, dogs, horses etc.). There are 6,000 images of each class with each image consisting of a colored $32 \times 32$ image. We perform private training and inference of AlexNet and VGG16 on this dataset.
    \item \textbf{Tiny ImageNet~\cite{tinyimagenet}: } Tiny ImageNet dataset consists of 100,000 training samples and 10,000 test samples with 200 different classes~\cite{tinyimagenet}. Each sample is cropped to a size of $64 \times 64 \times 3$. We perform private training and inference of AlexNet and VGG16 on this dataset.
\end{enumerate}

\vspace{-0.6cm}
\section{Networks}\label{app:networks}
\vspace{-0.2cm}
We evaluate \falcon on the following popular deep learning networks. We select these networks based on the varied range of model parameters and different types of layers used in the network architecture. The first three networks are purposely selected to perform performance comparison of \falcon with prior work that evaluated on these models. The number of layers that we report include only convolutional and fully connected layers. 
\iffullversion We provide a detailed configuration for each of the networks in Appendix~\ref{app:net_arch}. \fi
We also note that we enable the exact same functionality as prior work with no further approximations. Our networks achieve an accuracy of 97.42\% on Network-A, 97.81\% on Network-B, 98.64\% on Network-C, and 99.15\% on LeNet -- similar to the accuracy obtained by SecureNN, SecureML, and ABY$^3$~\cite{securenn, secureml, aby3}.

\begin{enumerate}[label=(\Alph*), itemsep=0pt, topsep=2pt, leftmargin=*]
    \item \textbf{Network-A: }This is a 3 layer fully-connected network with ReLU activation after each layer as was evaluated in SecureML~\cite{secureml}\iffullversion (see Figure~\ref{fig:secureml})\fi. This is the smallest network with around 118K parameters. %118282

    \item \textbf{Network-B: }This network is a 3 layer network with a single convolution layer followed by 2 fully-connected layers and ReLU activations. This architecture is chosen from Chameleon~\cite{chameleon} with approximately 100K parameters\iffullversion (see Figure~\ref{fig:chameleon})\fi. %99240

    \item \textbf{Network-C: }This is a 4 layer network with 2 convolutional and 2 fully-connected layers selected from prior work MiniONN~\cite{minionn}. This network uses Max Pooling in addition to ReLU layer and has around 10,500 parameters in total\iffullversion (shown in Figure~\ref{fig:minionn})\fi.% gives the detailed architecture. %33935

    \item \textbf{LeNet: }This network, first proposed by LeCun et al.~\cite{lenet} was used in automated detection of zip codes and digit recognition~\cite{zipcodes}. The network contains 2 convolutional layers and 2 fully connected layers with 431K parameters\iffullversion (shown in Figure~\ref{fig:lenet})\fi. %431080

    \item \textbf{AlexNet: } AlexNet is the famous winner of the 2012 ImageNet ILSVRC-2012 competition~\cite{alexnet}. It has 5 convolutional layers and 3 fully connected layers and uses batch norm layer for stability, efficient training and has about 60 million parameters\iffullversion (see Figure~\ref{fig:alexnet})\fi. \falcon is the first private deep learning framework that evaluates AlexNet because of the support for batch norm layer in our system.

    \item \textbf{VGG16: }The last network which we implement is called VGG16, the runner-up of the ILSVRC-2014 competition~\cite{vgg16}.  VGG16 has 16 layers and has about 138 million parameters\iffullversion (see Figure~\ref{fig:VGG16})\fi.
\end{enumerate}

%\vspace{-0.6cm}
%\section{Network Architectures}
%\label{app:net_arch}
%\vspace{-0.6cm}
%
%\input{Figures/SecureML}
%\input{Figures/MiniONN}
%\input{Figures/VGG16}
%\input{Figures/LeNet}
%\input{Figures/AlexNet}
%\input{Figures/Sarda}

\section{Functionality Descriptions}\label{app:functionalities}\vspace{-2mm}

\begin{Boxfig}{Ideal functionality for \protx{PC}}{FPC}{$\Funct{PC}$}
\footnotesize
\begin{description}

% \item[] The functionality receives inputs $\left\{ \Share{x_i}{\prm} \right\}_{i = 1}^{\ell}, r$
% \item[] Reconstruct bits $x_i$ and if reconstruction fails send $\Abort$ to the adversary. 
% \item[] Compute $x = \sum x_i \cdot 2^i$ and $b = (x > r)$
% \item[] Generate random shares of $b$ consistent with the adversarial outputs.
\item[Input: ]The functionality receives inputs $\{\Share{x_i}{\prm}\}_{i = 1}^{\ell}, r$

\item[Output: ]Compute the following 
\begin{enumerate}
    \item Reconstruct bits $x_i$ and $x = \sum x_i \cdot 2^i$
    \item Compute $b = (x \geq r)$
    \item Generate random shares of $b$ and send back to the parties
\end{enumerate}
\end{description}
\end{Boxfig}

\begin{Boxfig}{Ideal functionality for \protx{WA}}{FWA}{$\Funct{WA}$}
\footnotesize
\begin{description}

\item[Input: ]The functionality receives inputs $\Share{a}{L}$. 

\item[Output: ]Compute the following 
\begin{enumerate}
    \item Compute $b = \wa{3}(a_1,a_2,a_3,L)$
    \item Generate random shares of $b$ and send back to the parties
\end{enumerate}
\end{description}
\end{Boxfig}

\begin{Boxfig}{Ideal functionality for \protx{ReLU}}{FRELU}{$\Funct{ReLU}$}
\footnotesize
\begin{description}

\item[Input: ]The functionality receives inputs $\Share{a}{L}$. 

\item[Output: ]Compute the following 
\begin{enumerate}
    \item Compute $b = \mathsf{ReLU}(a_1+a_2+a_3 \pmod{L})$
    \item Generate random shares of $b$ and send back to the parties
\end{enumerate}
\end{description}
\end{Boxfig}

\begin{Boxfig}{Protocols for generating various pre-processing material}{fig:prot-prep}{\protx{Prep}}
\footnotesize
\begin{description}

\item[Usage: ]This is used to generate pre-processing material required for the online protocol. %We use a state-of-the-art protocol for the same adversarial model -- ABY$^3$ for generating this. 

\item[Setup:] This step will have to be done only once.
    \begin{enumerate}[noitemsep]
        \item Each party $P_i$ chooses a random seed $k_i$
        \item Send this random seed to party $P_{i+1}$
    \end{enumerate}
    
\item[Common randomness:] Let $F$ be any seeded PNRG. Then 3-out-of-3 and 2-out-of-3 common randomness described in Section~\ref{subsec:basicops} can be generated as follows:
    \begin{enumerate}[noitemsep]
        \item $\alpha_i = F_{k_i}(\mathsf{cnt}) - F_{k_{i-1}}(\mathsf{cnt})$ and  $\mathsf{cnt}$++
        \item $(\alpha_i, \alpha_{i-1}) = (F_{k_i}(\mathsf{cnt}),  F_{k_{i-1}}(\mathsf{cnt}))$ and $\mathsf{cnt}$++
    \end{enumerate}    

\item[Truncation Pair: ]Generate truncation pair $\share{r}, \share{r'} = \share{r/2^d}$. 
\begin{enumerate}
    \item Run protocol $\protx{trunc2}$ from~\cite{aby3} (Figure~3)
\end{enumerate}

\item[Correlated randomness for Private Compare: ]Generate correlated randomness required for $\protx{PC}$
    \begin{enumerate}
        \item Sample random bit $\Share{b}{2}$
        \item Use bit injection from~\cite{aby3} $\Share{b}{2} \rightarrow \Share{b}{\prm}$
        \item Sample random values $m_1, \hdots m_k \in \Z_{\prm}$.
        \item Compute and open $m_1^{\prm-1}, \hdots, m_k^{\prm-1}$.
        \item Remove openings that equal 0 and queue openings that equal 1. Note that this computation takes $\lceil \log_2 \prm \rceil$ rounds and can be amortized for efficiency (by setting a large value of $k$).
    \end{enumerate}
    
\item[Correlated randomness for Wrap$_3$: ]Generate correlated randomness required for $\protx{WA}$
    \begin{enumerate}
        \item Sample random bits $\Share{r_i}{2}$ for $i \in [\ell]$
        \item Perform bit composition from~\cite{aby3} to get $\Share{r_i}{\ring}$
        \item Use bit injection from~\cite{aby3} $\Share{r_i}{2} \rightarrow \Share{r_i}{\prm}$
        \item Use the optimized full adder $\mathsf{FA}$ to compute the final carry bit. Note that this bit is precisely $\wa{3}(\cdot)$ 
    \end{enumerate}
    
\item[Correlated randomness for ReLU: ]Generate correlated randomness required for $\protx{ReLU}$
    \begin{enumerate}
        \item Sample random bit $\Share{b}{2}$
        \item Use bit injection from~\cite{aby3} $\Share{b}{2} \rightarrow \Share{b}{\ring}$
    \end{enumerate}

\item[Correlated randomness for Maxpool and Division: ]No additional correlated randomness necessary other than that used in their subroutines.
\end{description}
\end{Boxfig}

\begin{Boxfig}{Ideal functionality for \protx{Maxpool}}{FMAX}{$\Funct{Maxpool}$}
\footnotesize
\begin{description}

\item[Input: ]The functionality receives inputs $\Share{a_1}{\ring}, \hdots \Share{a_n}{\ring}$.

\item[Output: ]Compute the following 
\begin{enumerate}
    \item Reconstruct $a_1,\hdots a_n$ and compute $k = \mathsf{argmax} \{a_1, \hdots a_n\}$.
    \item Set $\bm{e_k} = \{e_1, e_2, \hdots e_n\}$ with $e_i = 0 ~ \forall i \neq k$ and $e_k = 1$. 
    \item Generate random shares of $a_k$ and $\bm{e_k}$ and send back to the parties.
\end{enumerate}
\end{description}
\end{Boxfig}

\begin{Boxfig}{Functionality for \protx{Pow}}{FPOW}{$\Funct{Pow}$}
\footnotesize
\begin{description}

\item[Input: ]The functionality receives inputs $\Share{b}{\ring}$ and an index $k \in \{0, 1,\hdots \ell-1\}$. 

\item[Output: ]Compute each bit of $\alpha$ sequentially as follows:
\begin{enumerate}
    \item Reconstruct $b$.% and send $\Abort$ to all parties if that fails. 
    \item Compute $\alpha$ such that $2^{\alpha-1} < b \leq 2^{\alpha}$
    \item If $k = 0$ send $\alpha[i]$ for $i \in \{\ell-1, \hdots , 0\}$ to all parties.
    \item If $k \neq 0$ send $\alpha[i]$ for $i \in \{\ell-1, \hdots , k\}$ to all parties and then $\Abort$.
    % \item For $i \in \{ \ell-1, \hdots , 0\}$
    % \begin{enumerate}
    %     \item Send $\alpha[i]$ to each party ($i^{\tiny{\mathsf{th}}}$-bit of $\alpha$). 
    %     \item Wait for $\Abort$ or $\mathsf{Ok}$ from the adversary.
    %     \item If $\Abort$, send $\Abort$ to all parties and abort else continue.
    % \end{enumerate}
\end{enumerate}
\end{description}
\end{Boxfig}

\begin{Boxfig}{Ideal functionality for \protx{Div}}{FDIV}{$\Funct{Div}$}
\footnotesize
\begin{description}

\item[Input: ]The functionality receives inputs $\Share{a}{\ring}, \Share{b}{\ring}$ and an index $k \in \{0, 1,\hdots \ell-1\}$. 

\item[Output: ]Compute the following 
\begin{enumerate}
	\item Reconstruct $a, b$.% and send $\Abort$ to all parties if that fails. 
    \item Compute $\alpha$ such that $2^{\alpha} \leq b < 2^{\alpha+1}$
    \item If $k = 0$ send $\alpha[i]$ for $i \in \{\ell-1, \hdots , 0\}$ to all parties.
    \item If $k \neq 0$ send $\alpha[i]$ for $i \in \{\ell-1, \hdots , k\}$ to all parties and then $\Abort$.
    \item Generate random shares of $a \cdot \mathsf{AppDiv}(b)$ and send back to the parties 
\end{enumerate}
\end{description}
\end{Boxfig}

\begin{Boxfig}{Ideal functionality for \protx{BN}}{FBN}{$\Funct{BN}$}
\footnotesize
\begin{description}

\item[Input: ]The functionality receives inputs $\Share{a_1}{\ring}, \hdots \Share{a_n}{\ring}$ and $\Share{\gamma}{\ring}, \Share{\beta}{\ring}$ and an index $k \in \{0, 1,\hdots \ell-1\}$. 

\item[Output: ]Compute the following 
\begin{enumerate}
    \item Reconstruct $a_1,\hdots a_n$ and compute $\mu$ and  $\sigma^2$ as given in Step 1,2 of Algorithm~\ref{algo:batchnorm}
    \item Set $b = \sigma^2 + \epsilon$ and compute $\alpha$ such that $2^{\alpha} \leq b < 2^{\alpha+1}$
    \item If $k = 0$ send $\alpha[i]$ for $i \in \{\ell-1, \hdots , 0\}$ to all parties.
    \item If $k \neq 0$ send $\alpha[i]$ for $i \in \{\ell-1, \hdots , k\}$ to all parties and then $\Abort$.
    \item Complete steps 4-8 of Algorithm~\ref{algo:batchnorm} and return random shares of the output.
\end{enumerate}
\end{description}
\end{Boxfig}

\end{document}